\documentclass[12pt]{article}

\usepackage[margin=1in]{geometry}
\usepackage{amsmath,amssymb,amsthm}
\usepackage[usenames,dvipsnames]{xcolor}
\usepackage{tikz}
\usetikzlibrary{decorations.markings}
\usetikzlibrary{calc}
\usepackage[inline]{enumitem}
\usepackage{float}
\usepackage{xspace,nicefrac}
\usepackage[labelsep=none,format=plain,justification=centering]{caption}
\usepackage{thmtools}
\usepackage{thm-restate}

\renewcommand*{\Pr}{\mathop{\mathrm{Prob}}}
\newcommand*{\E}{\mathrm{\mathbf{E}}}

\newtheorem{theorem}{Theorem}
\newtheorem{lemma}{Lemma}

\newtheorem{corollary}{Corollary}

\def\nfrac#1#2{{\textstyle\frac{#1}{#2}}}
\def\dfrac#1#2{\lower0.15ex\hbox{\large$\frac{#1}{#2}$}}

\newcommand{\hide}[1]{}
\newcommand{\ignore}[1]{}

\newcommand{\csg}[1]{{\color{red}\ [CSG:~#1]}}

\newcommand{\tm}{make move\xspace}  
\newcommand{\tb}{break move\xspace} 
\newcommand{\mk}{{\tt{make}}}
\newcommand{\bk}{{\tt{break}}}
\newcommand{\tpr}{\ensuremath{\overline{C}_6}\xspace}
\newcommand{\ktt}{\ensuremath{K_{3,3}}\xspace}
\newcommand{\G}{{\ensuremath{\cal G}}\xspace}
\newcommand{\M}{{\ensuremath{\cal M}}\xspace}
\newcommand{\K}{{\ensuremath{\cal K}}\xspace}
\newcommand{\A}{{\ensuremath{\cal A}}\xspace}
\newcommand{\B}{{\ensuremath{\cal B}}\xspace}

\newcommand{\D}{\Delta}
\newcommand{\es}{\emptyset}
\def\d{\delta}
\def\a{\alpha}
\def\b{\beta}
\def\e{\epsilon}

\def\g{\gamma}

\newcommand{\brac}[1]{\left(#1\right)}
\newcommand{\bfrac}[2]{\left(\frac{#1}{#2}\right)}
\def\E{\mbox{{\bf E}}}

\def\Pr{\mbox{{\bf Pr}}}

\definecolor{brown}{cmyk}{0, 0.72, 1, 0.45}
\definecolor{grey}{gray}{0.5}

\allowdisplaybreaks

\tikzset{
  vb/.style     =
  { circle,
    inner sep       = 1pt,
    fill            = black,
    minimum size    = 0.75mm,
    draw
  },
  vw/.style     =
  { circle,
    inner sep       = 1pt,
    fill            = white,
    minimum size    = 0.75mm,
    draw
  },
  class/.style      =
  { rectangle,
    minimum size    = 6mm,
    rounded corners = 1mm,
    inner sep       = 3pt,
    fill            = none,
    draw
  }
}
\tikzset{every picture/.style={line width=1pt}}
\tikzset{empty/.style={rectangle,draw=none,fill=none}}
\tikzset{every node/.style={label distance=-2pt,font=\small}}

\newcommand*\tabrule[1][]{%
   \if\relax\detokenize{#1}\relax
     \rule{\linewidth}{\arrayrulewidth}%
   \else
     \rule{#1}{\arrayrulewidth}%
   \fi
}

\date{\today}

\title{Triangle-creation processes on cubic graphs
\thanks{
Research supported by UK EPSRC grant EP/M005038/1,
``Randomized algorithms for computer networks''.}
}
\author{
Colin Cooper
\\
{\small Department of Informatics}\\[-0.8ex]
{\small Kings College}\\[-0.8ex]
{\small London WC2R 2LS, U.K.}\\
{\small \texttt{colin.cooper@kcl.ac.uk}}\\
\and
Martin Dyer\\
{\small School of Computing}\\[-0.8ex]
{\small University of Leeds}\\[-0.8ex]
{\small Leeds LS2 9JT, U.K.}\\
{\small \texttt{m.e.dyer@leeds.ac.uk}}\\
\and
Catherine Greenhill\\
{\small School of Mathematics and Statistics}\\[-0.8ex]
{\small UNSW Sydney}\\[-0.8ex]
{\small NSW 2052, Australia}\\
{\small \texttt{c.greenhill@unsw.edu.au}}
}

\begin{document}

\maketitle

\begin{abstract}
An edge switch is an operation which makes a local change in a graph while maintaining the degree of every vertex.
We introduce a switch move,  called a triangle switch, which creates or deletes at least one triangle at each step.
Specifically, a make move is a triangle switch which chooses a path $zwvxy$ of length 4 and replaces
it by a triangle $vxwv$ and an edge $yz$, while a break move performs the reverse operation.
We consider various Markov chains which perform random triangle switches, and assume that every possible make
or break move has positive probability of being performed.

Our first result is that any such Markov chain is irreducible on the set of all 3-regular graphs with vertex set
$\{1,2,\ldots, n\}$.
For a particular, natural Markov chain of this type,
 we obtain a non-trivial linear upper and lower bounds on the number of triangles in the long run.
These bounds are almost surely obtained in linear time, irrespective of the
starting graph.
\end{abstract}

\section{Introduction}\label{intro:sec}

In the applied field of social networks, the existence of triangles is
seen as an indicator of mutual friendships~\cite{GKM,JGN}.  However, many random graph models,
or processes for producing random graphs,  tend to produce graphs with few triangles. We introduce and analyse random processes based on Markov chains,
which are designed to favour graphs with many triangles.
We restrict our attention to cubic (3-regular) graphs, as the questions we seek to answer are already challenging in this setting.

Let $\mathcal{G}_n$ be the set of simple 3-regular
graphs on the vertex set $[n]=\{1,2,\ldots, n\}$.  Assume that $n\geq 4$
and that $n$ is even, as otherwise $\mathcal{G}_n=\emptyset$.
A random 3-regular graph on $n$ vertices is a graph sampled uniformly at
random (u.a.r.) from $\G_n$.  The number of cycles of length $\ell\geq 3$
in a random 3-regular graph is asymptotically Poisson with expected value
$2^{\ell}/(2\ell)$, see for example Bollob{\' a}s~\cite{Bo}.
A \emph{triangle} is a 3-cycle, and the expected number of triangles in a
random 3-regular graph is asymptotically equal to $4/3$.  We are interested
in random processes which produce graphs that contain many more triangles
than we would expect from sampling uniformly at random.

Of course, some elements of $\G_n$ contain many more triangles than the average.
For example, suppose that $n\equiv 0\!\pmod 4$ and let $\K_n$ be the set of all graphs
consisting of
$n/4$ components, each of which is isomorphic to $K_4$.
Then every graph in $\K_n$ has $n$ triangles, and this is the largest possible number of
triangles in a 3-regular graph on $n$ vertices.
Note that, while any two graphs in $\K_n$ are isomorphic, we have
\[ |\K_n| = \frac{n!}{(4!)^{n/4}} \sim\sqrt{2 \pi n}\, \left(\frac{n}{e(24)^{1/4}}\right)^n>(n/7)^n.\]
But $\K_n$ is an exponentially small proportion of $\G_n$ since
\[
|\G_n| \sim \sqrt 2e^{-2} \, \bfrac{3^{1/2} n^{3/2}}{2 e^{3/2}}^n
\]
(see~\cite{Bo}). Therefore, when sampling from $\G_n$ uniformly at random, we are
exponentially unlikely to see a graph from $\K_n$.

Since the introduction of the random graph models $G_{n,m}$ by Erd{\H o}s and R{\' e}nyi \cite{ER}, and $G_{n,p}$ by Gilbert \cite{Gil}, there has been an ongoing study of models of random graphs and their properties by many authors.
Generally speaking, the graph properties follow naturally as a
consequence of the generative model, or are introduced artificially into the generative process.

One straightforward but artificial way to generate 3-regular graphs with a fixed number of triangles is as follows. Starting with an $n$-vertex 3-regular graph, choose a random subset of $m$ vertices, and replace each of these vertices with a triangle. This gives a 3-regular graph on $n+2m$ vertices with approximately $m$ triangles.

A more natural method to generate regular graphs with a large number of triangles is to use
local transformations.  The method of local transformations makes modifications of the graph structure to alter the density of triangles in the long run. If valid transformations are accepted with variable probabilities, this is often called \emph{Metropolis sampling}.

An established approach to the uniform generation of regular graphs is using
local edge transformations known as switches (see for example~\cite{CoDyGr07,KTV,MES}).
 A pair of edges $xy$, $wz$ of graph $G$ are chosen u.a.r.\ and replaced with a uniformly
chosen perfect matching of the vertices $\{x,y,w,z\}$. If the resulting graph $G'$ is not simple then the move is rejected.  See Fig.~\ref{fig:switch}.
\begin{figure}[H]
\begin{center}
\begin{tikzpicture}[scale=0.74]
\draw [thick,-] (4,0) -- (2,0);\draw (2,2) -- (4,2);
\draw [fill] (2,0) circle (0.1);
\draw [fill] (2,2) circle (0.1);
\draw [fill] (4,2) circle (0.1); \draw [fill] (4,0) circle (0.1);
\node [above] at (2,2.1)  {$x$}; \node [below] at (2,-0.1)  {$w$};
\node [above] at (4,2.1)  {$y$}; \node [below] at (4,-0.1)  {$z$};
\draw  (6,1.3)edge[line width=1.5pt,<->](8,1.3);
\draw (7,1.3)node[empty,label=above:switch] {} ;
\begin{scope}[shift={(8.5,0)}]
\draw [thick,-] (2,2) -- (2,0);
\draw [thick,-] (4,0) -- (4,2);
\draw [fill] (2,0) circle (0.1);
\draw [fill] (2,2) circle (0.1);
\draw [fill] (4,2) circle (0.1); \draw [fill] (4,0) circle (0.1);
\node [above] at (2,2.1)  {$x$}; \node [below] at (2,-0.1)  {$w$};
\node [above] at (4,2.1)  {$y$}; \node [below] at (4,-0.1)  {$z$};
\end{scope}
\end{tikzpicture}
\caption{\;\; A switch}\label{fig:switch}
\end{center}
\end{figure}
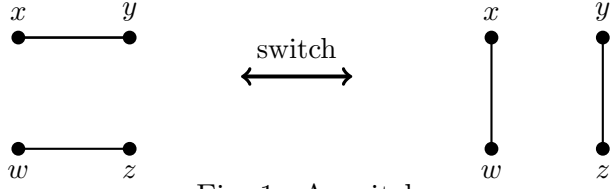
Here we use the notation $xy$ as a shorthand for the edge $\{x,y\}$.
On some occasions it is convenient to temporarily assign an orientation
to an edge: in these situations, $xy$ will denote the oriented edge $(x,y)$.

If we restrict to a subset of switches, we can ensure that every switch will change
the set of triangles in the graph.  Fig.~\ref{proc:fig1} illustrates a \emph{triangle
switch} which can be used to
make or break triangles.
\begin{figure}[H]
\begin{center}
\begin{tikzpicture}[scale=0.75]
\draw [very thick,-] (4,0) -- (2,0) -- (0,1) -- (2,2) -- (4,2);
\draw [fill] (0,1) circle (0.1); \draw [fill] (2,0) circle (0.1);
\draw [fill] (2,2) circle (0.1);
\draw [fill] (4,2) circle (0.1); \draw [fill] (4,0) circle (0.1);
\node [left] at (-0.1,1)  {$v$};
\node [above] at (2,2.1)  {$x$}; \node [below] at (2,-0.1)  {$w$};
\node [above] at (4,2.1)  {$y$}; \node [below] at (4,-0.1)  {$z$};
\draw  (6,1.3)edge[line width=1.5pt,->](8,1.3);
\draw (7,1.3)node[empty,label=above:make] {} ;
\draw  (6,0.7)edge[line width=1.5pt,<-](8,0.7);
\draw (7,0.7)node[empty,label=below:break] {} ;
\begin{scope}[shift={(10,0)}]
\draw [very thick,-] (0,1) -- (2,2) -- (2,0) -- (0,1);
\draw [very thick,-] (4,0) -- (4,2);
\draw [fill] (0,1) circle (0.1); \draw [fill] (2,0) circle (0.1);
\draw [fill] (2,2) circle (0.1);
\draw [fill] (4,2) circle (0.1); \draw [fill] (4,0) circle (0.1);
\node [left] at (-0.1,1)  {$v$};
\node [above] at (2,2.1)  {$x$}; \node [below] at (2,-0.1)  {$w$};
\node [above] at (4,2.1)  {$y$}; \node [below] at (4,-0.1)  {$z$};
\end{scope}
\end{tikzpicture}
\end{center}
\caption{\;\; A triangle switch}\label{proc:fig1}
\end{figure}
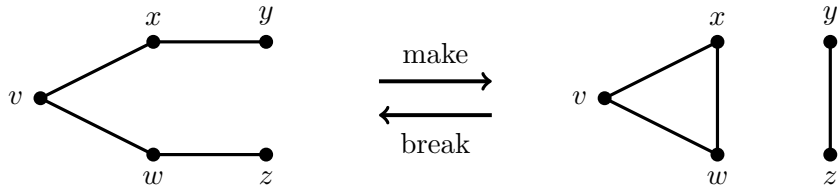
Suppose that the 4-path $yxvwz$ is present in the graph and the edges
$xw$, $yz$ are absent.  Then the {\em make triangle move}
consists of deleting the edges $xy$, $wz$ and replacing them with
edges $xw$, $yz$, forming the triangle $vxwv$.
The make triangle move is illustrated in Fig.~\ref{proc:fig1}, reading from left to right.
We denote this operation by $\mk(yxvwz)$. As the make move depends only
on the 4-path and not its orientation, we treat $\mk(yxvwz)$ and
$\mk(zwvxy)$ as the same make move.


Next, suppose that the triangle $vxwv$ and disjoint edge $yz$ are present
in the graph, and edges $xy$, $wz$ are absent.  Here we treat both
the edge $xw$ of the triangle and the disjoint edge $yz$ as oriented edges.
The {\em break triangle move} deletes the edges $xw$, $yz$ and replaces them
with the edges $xy$, $wz$.  This destroys the triangle $vxwv$, but may create
other triangles.
The break triangle move is illustrated in Fig.~\ref{proc:fig1}, reading from right to
left.  We denote this operation by $\bk(vxw,yz)$.
Again, by symmetry we will treat $\bk(vxw,yz)$ and $\bk(vwx,zy)$ as the same
break move.

As we require the resulting graph to be simple, we will reject any make or break move
which would create a repeated edge or where the vertices $v,x,y,w,z$ are not distinct.
If a move is rejected then no change is made at that step.

A Markov chain
$\mathcal{M}$ on $\mathcal{G}_n$ is called a \emph{triangle switch chain}
if its transition matrix $P$ satisfies $P(G,G')>0$ whenever $G$ and $G'$ differ
by a single (make or break) triangle switch.
Two general types of  Markov chains based on triangle switches
are as follows at each step.
\begin{itemize}
\item With probability $p$ choose a \tm, else (with probability $1-p$)
choose a \tb.  Carry out a u.a.r.\ move of the chosen type.
\item
Sample a vertex $v$ and two neighbours $w,x$. If there is no edge $wx$  then
with probability $p$ perform a \tm.
If the edge $xw$ is present then with probability $q$ perform a \tb.
\end{itemize}
In this paper we prove that switch processes which allow all moves of the
type shown in Fig.~\ref{proc:fig1}  are irreducible
on the set of simple 3-regular graphs $\G_n$, and estimate (for two particular triangle switch chains) the density of triangles produced in relation to the proportion of make moves.

Whether or not a Markov chain is irreducible depends only on the set of
transitions which have nonzero probability, and not on the values of these
probabilities.
So, in fact, the irreducibility of these chains is independent of exactly how {\tm}s and {\tb}s are chosen, provided all valid moves at any vertex can occur with nonzero probability. Thus a proof of irreducibility for either of the above chains extends to the other, and to all similar chains we might devise.

Our first result, stated below, is proved in Section~\ref{ergo:sec}.

\begin{restatable}{theorem}{TheoremI}\label{TH1}
Any triangle switch chain on the space of simple 3-regular graphs $\G_n$ is irreducible.
\end{restatable}

\subsection{A variety of chains and an alternative Metropolis process}\label{MetPro}

Given a graph $G=(V,E)$, a vertex $v \in V$ (of degree at least two) defines a set of triples  $(v; w,x)$ consisting of paths $wvx$ of length two, where $w,x$ are distinct neighbours of $v$ ($w,x \in N(v)$).
Let $S(G)$ be the set of all such triples in $G$. Let $M(G)$ be the set of those triples $(v;w,x)\in S(G)$ for which $wx \not \in E$. Each such triple is a
potential site for a \tm, as  illustrated on the left of Fig.~\ref{proc:fig1}.
Next let $B(G) = S(G) \setminus M(G)$ be the set of all
triples $(v;w,x)\in S(G)$ for which  $wx \in E$.  Each triple in $B(G)$ is
a potential site for a \tb, as  illustrated on
the right of Fig.~\ref{proc:fig1}.

As mentioned previously, there are several ways to design a Markov chain based on triangle switches. The first approach is to decide the type of move (make/break) and then sample from the available moves of that type.  One such chain is described below.

\begin{center}
\fbox{
\parbox{0.9\textwidth}{
{\bf Chain~O}.
Repeat for $R$ steps:\\[0.2ex]
With probability $p$, \\[0.2ex]
\phantom{XX} $\bullet$ sample a triple $(v;w,x)$ u.a.r.\ from $M(G)$, \\[0.2ex]
\phantom{XX} $\bullet$ choose $y\in N_G(x)\setminus \{v\}$ and $z\in N_G(w)\setminus\{v\}$ u.a.r.,\\[0.2ex]
\phantom{XX} $\bullet$ perform $\mk(yxvwz)$ if valid;\\[0.2ex]
otherwise (with probability $1-p$),\\[0.2ex]
\phantom{XX} $\bullet$  sample a triple $(v;w,x)$ u.a.r.\ from $B(G)$,\\[0.2ex]
\phantom{XX} $\bullet$ choose an oriented edge $yz\in E(G)$ u.a.r.,\\[0.2ex]
\phantom{XX} $\bullet$ perform $\bk(vxw,yz)$ if valid.}
}
\end{center}

In Chain~O, at each step the probability of attempting a make move is $p$, and
the probability of attempting a break move is $1-p$,
regardless of the number of triangles which contain the vertex $v$.
Chain~O is a time-reversible Markov chain  which is somewhat like a biased
random walk on the number of triangles.
This analogy is not exactly correct, as proposed moves may be invalid, or
may add or remove more than one triangle in some cases.

To implement Chain~O efficiently we can proceed as follows.
For the initial graph $G$, form the sets $M(G)$, $B(G)$  in $O(n)$ time.  After each transition of the chain,
update these sets in $O(1)$ time.  The running time of Chain~O is then
determined mainly by the number of steps $R$.

\bigskip

An alternative approach is to first  choose a vertex, and then choose a move at that vertex in some way. We study two chains with this structure.
In the first chain, the probability of a move depends only on
its type (make/break), while in the second chain these probabilities
depend on the number
of triangles in the neighbourhood of the chosen vertex.

\begin{center}
\fbox{
\parbox{0.9\textwidth}{
{\bf Chain~I}. Choose u.a.r.\ a vertex  $v$ and an ordered pair of
distinct edges
$(vw,vx)$.
Depending on the presence (or absence) of the edge $wx$,
choose a \tb with probability $q$ (resp. a \tm with probability $p$).
This chain is defined in more detail in Section~\ref{proc1:sec}, see Fig.~\ref{fig1}.

\medskip

{\bf Chain~II}.  Choose a vertex  $v$ u.a.r.\ and let $\D_v$ be the number of triangles  at $v$. Choose a  \tb with probability $\D_v/3$, or a \tm otherwise.
Moves are chosen u.a.r.\ from among the valid moves of that type at vertex $v$.
See Section \ref{proc2:sec}, Fig.~\ref{fig2} for more detail.
}}
\end{center}

Chain~II differs from Chain~I in that (i) in Chain~I the existence of the edge $wx$ determines the move type, and (ii) Chain~II always attempts to perform a move
of the chosen type, whereas Chain~I only attempts to perform the chosen move with probability $p$ or $q$. It follows from the proof of Theorem~\ref{TH1}
(see Lemma~\ref{ergo:lem10}) that, given a vertex $v$ which is not in a $K_4$, there is always a valid \tm at $v$ which inserts a triangle at $v$.
(However, the same move may break some other triangle at $v$.)

\vspace{0.2 in}

We say that an event $\cal E$ {\em occurs in the long run}, if
there exists a positive constant $C$ such that for all $t\geq C n$,
the number of steps at which $\cal E$ does not hold is $o(t)$ almost surely.
For Chain II, we  give bounds on the long run number of triangles.
The following theorem is proved in Section~\ref{number:sec}.

\begin{theorem}\label{TH2}
Let $\a=0.09$ and $\b=0.63$.
There exists a positive constant $C$ such that the
following statements hold for 3-regular graphs $G(t)$ sampled from Chain II.
\begin{enumerate}
 \item[\emph{(i)}]
 After $t=C n$ steps of Chain II, w.h.p.\ the number of triangles in the current graph $G(t)$ is at least $\a n$, and at most $\b n$ independent of the starting graph $G(0)$.
\item[\emph{(ii)}] The long run number of triangles is at least $\a n$,  and at most $\b n$.
\end{enumerate}
\end{theorem}
Adapting the proof of Theorem~\ref{TH2}
we  obtain the following corollary.
\begin{corollary}\label{coTH2}
The long run number of triangles  in Chain~I is at least
$(1-\epsilon)np/(72-63p)$.
\end{corollary}

The three chains introduced in this section are all reversible, and are irreducible by Theorem~\ref{TH1}.
Chains~O and I are aperiodic, assuming that the probabilities $p,q\in (0,1)$,
and the implementation of Chain~II described in Section~\ref{proc2:sec} is also aperiodic.  Therefore,
each of these chains has a unique stationary distribution on ${\cal G}_n$.
In principle by varying $p, q$ we can to use Chain~O or Chain~I to sample 3-regular
graphs with (approximately) a specified number of triangles,
with reasonable accuracy.  We now discuss this further for Chain~I.

\bigskip

Denote the stationary probabilities of $G$ and $G'$
by $\Pr(G)$, $\Pr(G')$, respectively.
We claim that
the implementation of Chain~I described in Section~\ref{proc1:sec}
satisfies the following detailed balance equations:
\begin{equation}
\label{balance}
\frac{p}{12n}\, \Pr(G)= \frac{q}{9n^2}\, \Pr(G').
\end{equation}

First suppose that the graph $G'\in\G_n$ can be obtained from the graph
$G\in \G_n$ only by performing $\mk(yxvwz)$ (recalling that we
treat this move and its mirror image, $\mk(zwvxy)$, as the same make move).
There are $12n$ ways of choosing the (unoriented) path $yxvwz$ for a \tm,
and $18n^2$ ways of choosing the oriented path $vxw$ (around a triangle)
and oriented edge $yz$ for a \tb, but then we must divide by 2 as
$\bk(vxw,yz)$ and $\bk(vwx;zy)$ give the same break move.
So there are $9n^2$ distinct ways to choose a break move. It follows that
(\ref{balance}) holds in this case.

If the transition from $G$ to $G'$ can arise from more than one make move
then (\ref{balance}) still holds.  For example,
suppose that $G$ contains a 6-cycle $vwzuyxv$.
Then the transition from $G$ to $G'$
arises in two distinct ways, as $\mk(yxvwz)$ and as $\mk(wzuyx)$
(again, identifying
each path with its mirror image).  Similarly, the
reverse transition from $G'$ to $G$ arises in two ways, as
$\bk(vxw,yz)$ and as $\bk(uyz,xw)$.
See also Fig.~\ref{fig:metro}, giving a situation in which the transition
from $G$ to $G'$ arises in four ways. A final possibility is that
$G$ to $G'$ can arise from three distinct make moves (similar to
Fig.~\ref{fig:metro} but with edge $ab$ absent, for example).

In Chain~I, if $q=1-p$ and $p$ is close to zero then {\tb}s will dominate.
The resulting stationary distribution will favour graphs with few triangles.
In particular, when $p=4/(3n+4)$ we have $\Pr(G')=\Pr(G)$.
Hence, when $p=4/(3n+4)$ and $q=1-p$, Chain~I simply generates random 3-regular graphs, which have $O(1)$ triangles. On the other hand, if $p$ is close to one then {\tm}s will dominate, so the stationary distribution will favour
graphs with many triangles.

\subsubsection{An alternative Metropolis process}

Consider next another natural process based on switches.
Fix a parameter $q\in (0,1)$ and let $\Delta(G)$ denote the number of triangles
in a graph $G$.
We choose a switch at random, as in the algorithm of \cite{CoDyGr07}.
If the proposed transition is from the current graph $G$ to $G'$ then the move is accepted with probability $q^{\Delta(G)-\Delta(G')+4}$, and otherwise we stay at $G$.
Note that $\Delta(G')\leq\Delta(G)+4$, with an extreme configuration
depicted in Fig.~\ref{fig:metro}.  Hence the expression
$q^{\Delta(G)-\Delta(G')+4}$ belongs to $(0,1]$, so it is indeed a probability.
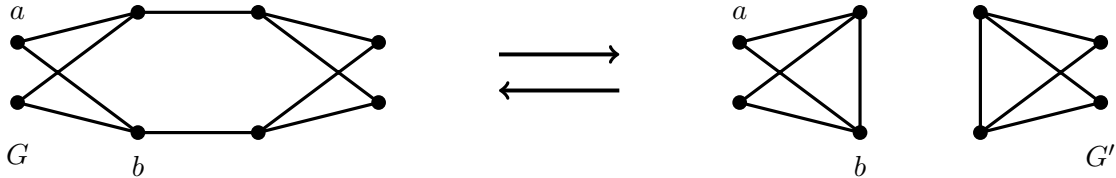
\begin{figure}[H]
\begin{center}
\begin{tikzpicture}[scale=0.8]
\draw [very thick,-] (6,0.5)--(4,0) -- (2,0) -- (0,0.5) -- (2,2) -- (4,2)--(6,0.5) (2,2)--(0,1.5)--(2,0) (4,2)--(6,1.5)--(4,0);
\draw [fill] (0,0.5) circle (0.1); \draw [fill] (0,1.5) circle (0.1); \draw [fill] (2,0) circle (0.1);
\draw [fill] (2,2) circle (0.1);
\draw [fill] (4,2) circle (0.1); \draw [fill] (4,0) circle (0.1);
\draw [fill] (6,0.5) circle (0.1) (6,1.5) circle (0.1);
\draw  (8,1.3)edge[line width=1.5pt,->](10,1.3);
\draw  (8,0.7)edge[line width=1.5pt,<-](10,0.7);
\node [below] at (0,0) {$G$};
\node [above] at (0,1.7) {$a$};
\node [below] at (2,-0.2) {$b$};
\begin{scope}[shift={(12,0)}]
\draw [very thick,-] (6,0.5)--(4,0) (2,0) -- (0,0.5) -- (2,2) (4,2)--(6,0.5) (2,2)--(0,1.5)--(2,0) (4,2)--(6,1.5)--(4,0) (2,0)--(2,2) (4,0)--(4,2) ;
\draw [fill] (0,0.5) circle (0.1); \draw [fill] (0,1.5) circle (0.1); \draw [fill] (2,0) circle (0.1);
\draw [fill] (2,2) circle (0.1);
\draw [fill] (4,2) circle (0.1); \draw [fill] (4,0) circle (0.1);
\draw [fill] (6,0.5) circle (0.1) (6,1.5) circle (0.1);
\node [below] at (6,0) {$G'$};
\node [above] at (0,1.7) {$a$};
\node [below] at (2,-0.2) {$b$};
\end{scope}
\end{tikzpicture}
\end{center}
\caption{\;\; A worst-case switch}\label{fig:metro}
\end{figure}
\noindent This is a Metropolis process with
detailed balance equations
\[  q^{\Delta(G)-\Delta(G')+4}\, \Pr(G)=q^{\Delta(G')-\Delta(G)+4}\, \Pr(G').\]
Letting $Q=\sum_{G\in \G_n} q^{-2\Delta(G)}$, these equations have solution
$\Pr(G)=q^{-2\Delta(G)}/Q$. Therefore, graphs with many triangles are more likely in the stationary distribution than those with few. It might be possible to show that this chain is rapidly mixing for small values of $q$  by modifying the methods of \cite{CoDyGr07}.

However, this chain is likely to be very slow, especially if $q$ is chosen so that graphs with linearly
many triangles have high stationary probability.
For this reason, we prefer to work with Chains O, I and II, which can produce graphs with $\Omega(n)$ triangles in $O(n)$ time.  However, we are unable to determine the exact stationary distribution, or bound the mixing time, of these chains.

\bigskip

The structure of the paper is as follows.
In Section \ref{ergo:sec}, we establish the irreducibility of the triangle switch chain, proving Theorem~\ref{TH1}.  In Section \ref{number:sec} we give a detailed definition of Chains~I and II. We also prove Theorem~\ref{TH2}, giving bounds on the long run expected number of triangles generated by Chain II, and use these bounds to prove Corollary \ref{coTH2} for Chain I.

We refer to the following three induced subgraphs throughout our proofs.
Many of our arguments are based on analysis of how the triangle switch
creates and destroys these induced substructures.
\begin{figure}[ht!]
\begin{center}
\begin{tikzpicture}[scale=1.5]
\draw (0,0) node[vb] (0) {} (1,0) node[vb] (1) {} (0.5,0.87) node[vb] (2) {} ;
\draw (0)--(1)--(2)--(0) ;
\draw (0.5,-0.5) node[empty] {triangle $K_3$} ;
\end{tikzpicture}
\hspace*{2cm}
\begin{tikzpicture}[xscale=1.67,yscale=1]
\draw (0,0) node[vb] (0) {} (1,0) node[vb] (1) {} (0.5,0.87) node[vb] (2) {} (0.5,-0.87) node[vb] (3) {} ;
\draw (0)--(2)--(1)--(3)--(0)--(1) ;
\draw (0.5,-1.7) node[empty] {diamond} ;
\end{tikzpicture}
\hspace*{2cm}
\begin{tikzpicture}[xscale=1,yscale=1]
\draw (0,0) node[vb] (0) {} (2,0) node[vb] (1) {} (1,1.73) node[vb] (2) {} (1,0.7) node[vb] (3) {} ;
\draw (0)--(1)--(2)--(0)--(3)--(1) (2)--(3) ;
\draw (1,-0.7) node[empty] {tetrahedron $K_4$}  ;
\end{tikzpicture}
\end{center}
\caption{}\label{intro:fig1}
\end{figure}

\section{Irreducibility}\label{ergo:sec}

Let $\G^*_n$ be the graph with vertex set $\G_n$ and with edges
defined as follows:
$G, G'\in\G_n$ are joined by an edge if there is a \tm or \tb
which takes $G$ to $G'$.
Then $\G^*_n$ is an undirected graph, which is the transition graph of any
triangle switch chain $\M_n$.
To show that $\M_n$ is irreducible,
we need to show that $\G^*_n$ is a connected graph.

More generally, given any $\mathcal{H} \subseteq \G_n$, let $\mathcal{H}^*$
be the vertex-induced subgraph of $\G^*_n$ induced by the set $\mathcal{H}$.

For now, we consider the case $n\equiv 0\!\pmod 4$.
First we will prove that there is a path from any $G\in \G_n$ to a graph $H\in \K_n$, where $\K_n$ is the class of labelled graphs
consisting of $n/4$ disjoint copies of
$K_4$. We may assume that $n\geq 8$, since otherwise $G$ is already a $K_4$.
Then we will show that any two graphs in $\K^*_n$ are connected by
a path in $\G^*_n$, completing the proof when $n\equiv 0\!\pmod{4}$.
Finally, we will show how to modify the proof for the remaining case $n\equiv 2\!\pmod 4$.

\begin{theorem}\label{ergo:thm1}
  If $n\equiv 0\!\pmod 4$ then there is a path in $\G^*_n$ from any $G\in\G_n$ to some $H\in \K_n$.
\end{theorem}

We will prove this using a sequence of lemmas which show that the process can always increase the number of triangles at some vertex, so long as a component remains which is not isomorphic to $K_4$.

\begin{lemma}\label{ergo:lem10}
Let $G\in\G_n$. If $v\in G$ has no triangle in its neighbourhood then there is an adjacent $G'$ in $\G^*_n$ which has a triangle in the neighbourhood of $v$.
\end{lemma}

\begin{proof}
The situation is depicted in Fig.~\ref{ergo:fig1}(a). The vertices $v,w,x,y$ must be distinct, with no edges between $w,x,y$, since $v$ is not in a triangle. The vertices $a,b,c,d,e,f$ must be distinct from $w,x,y$, but not necessarily from each other, and from $v$, since $v$ is not in a triangle. Note also that  since $G$ is a simple graph, the pair $a,b$ must be distinct, and similarly the pair $c,d$ and the pair $e,f$.

We use the notation $u\cdot v$ to indicate that vertices $u$ and $v$ coincide (that is, $u=v$). The following cases concern the extent to which vertices
$a,b,c,d,e,f$ coincide.  For clarity, when describing an edge or path
involving a vertex of the form $u\cdot v$, we will put parentheses around this
vertex.

\begin{figure}[ht]
\begin{center}
\begin{tikzpicture}[scale=1]
\draw (0,0) node[vb,label=below:$v$] (v) {} (-30:1) node[vb,label=above:$y$] (y) {} (90:1) node[vb,label=right:$w$] (w) {} (210:1) node[vb,label=above:$x$] (x) {};
\draw (80:2) node[vb,label=right:$a$] (a) {} (100:2) node[vb,label=left:$b$] (b) {} (200:2) node[vb,label=left:$c\,$] (c) {}(220:2) node[vb,label=left:$d$] (d) {} (-20:2) node[vb,label=right:$f$] (f) {} (-40:2) node[vb,label=right:$e$] (e) {} ;
\draw (v)--(w) (v)--(x) (v)--(y) (b)--(w)--(a) (c)--(x)--(d) (e)--(y)--(f) ;
\draw (0,-2) node[empty] {(a)} ;
\end{tikzpicture}
\hspace*{5mm}
\begin{tikzpicture}[scale=1]
\draw (0,0) node[vb,label=below:$v$] (v) {} (-30:1) node[vb,label=above:$y$] (y) {} (90:1) node[vb,label=right:$w$] (w) {} (210:1) node[vb,label=above:$x$] (x) {};
\draw (80:2) node[vb,label=right:$a$] (a) {} (100:2) node[vb,label=left:$b$] (b) {} (200:2) node[vb,label=left:$c\,$] (c) {}(220:2) node[vb,label=left:$d$] (d) {} (-20:2) node[vb,label=right:$f$] (f) {} (-40:2) node[vb,label=right:$e$] (e) {} ;
\draw (v)--(w) (v)--(x) (v)--(y) (b)--(w)--(a) (c)--(x)--(d) (e)--(y)--(f) ;
\draw (a)--(f) (a)--(c) ;
\draw (0,-2) node[empty] {(a$'$)} ;
\end{tikzpicture}
\hspace*{5mm}
\begin{tikzpicture}[scale=1]
\draw (0,0) node[vb,label=below:$v$] (v) {} (-30:1) node[vb,label=above:$y$] (y) {} (90:1) node[vb,label=right:$w$] (w) {} (210:1) node[vb,label=above:$x$] (x) {};
\draw (80:2) node[vb,label=right:$a$] (a) {} (100:2) node[vb,label=left:$b$] (b) {}
(210:2) node[vb,label=left:$c\,$] (c) {} (270:1.1) node[vb,label=below:$d\cdot e$] (de) {}
(-30:2) node[vb,label=right:$f$] (f) {} ;
\draw (v)--(w) (v)--(x)--(de) (v)--(y)--(de) (b)--(w)--(a)  (x)--(c) (y)--(f) ;
\draw (a)--(f) (a)--(c) ;
\draw (0,-2) node[empty] {(b)} ;
\end{tikzpicture}\\[2ex]
\begin{tikzpicture}[scale=0.8]
\draw (0,0) node[vb,label=below:$v$] (v) {} (-30:1) node[vb,label=above:$y$] (y) {} (90:1) node[vb,label=right:$w$] (w) {} (210:1) node[vb,label=above:$x$] (x) {};
\draw (80:2) node[vb,label=right:$a$] (a) {} (100:2) node[vb,label=left:$b$] (b) {}
(210:2) node[vb,label=below:$d\cdot e$] (de) {}
(-30:2) node[vb,label=below:$c\cdot f$] (cf) {} ;
\draw (v)--(w) (v)--(x)--(de) (v)--(y)--(de) (b)--(w)--(a)  (x)--(cf) (y)--(cf) ;
\draw (a)--(cf) (a)--(de) ;
\draw (0,-2) node[empty] {(c)} ;
\end{tikzpicture}
\hspace*{1cm}
\begin{tikzpicture}[scale=0.8]
\draw (0,0) node[vb,label=below:$v$] (v) {} (-30:1) node[vb,label=right:$y$] (y) {} (90:0.8) node[vb,label= above:$w$] (w) {} (210:1) node[vb,label=left:$x$] (x) {};
\draw (55:1.6) node[vb,label=right:$a\cdot f$] (af) {} (125:1.6) node[vb,label=left:$b\cdot c$] (bc) {}
(270:1.2) node[vb,label=below:$d\cdot e$] (de) {} ;
\draw (v)--(w) (v)--(x)--(de) (v)--(y)--(de) (bc)--(w)--(af)   (y)--(af) ;
\draw (x)--(bc) (y)--(af) (bc)--(af) ;
\draw (0,-2.3) node[empty] {(d)} ;
\end{tikzpicture}
\hspace*{1cm}
\begin{tikzpicture}[scale=0.8]
\draw (0,0) node[vb,label=below:$v$] (v) {} (-30:1) node[vb,label=right:$y$] (y) {} (90:1) node[vb,label=left:$w$] (w) {} (210:1) node[vb,label=left:$x$] (x) {};
\draw (60:2) node[vb,label=right:$a\cdot c\cdot f$] (a) {} (100:2) node[vb,label=left:$b$] (b) {} (60:2) node[vb] (c) {}(220:2) node[vb,label=left:$d$] (d) {} (-40:2) node[vb,label=right:$e$] (e) {} (60:2) node[vb] (f) {} ;
\draw (v)--(w) (v)--(x) (v)--(y) (b)--(w)--(a) (c)--(x)--(d) (e)--(y)--(f) ;
\draw(a)--(f) (a)--(c) ;
\draw (0,-2) node[empty] {(e)} ;
\end{tikzpicture}
\end{center}
\caption{}\label{ergo:fig1}
\end{figure}
Vertex $a$ is treated as a distinguished vertex. We separate out various cases to examine the effect of identification of vertices in $\{a,b,c,d,e,f\}$. In each case, if $N(v)=\{w,x,y\}$ is triangle-free then we can insert at least one of the edges $xw, wy, xy$.
The following cases are exhaustive, as all others are excluded by 3-regularity.
\smallskip

\begin{enumerate}[itemsep=0pt,topsep=0pt,label=(\alph*)]
  \item \emph{None of $\{a,b,c,d,e,f\}$ coincide.}\ See Fig.~\ref{ergo:fig1}(a).
Consider any 2-path from $v$, say $vwa$. There are four possible {\tm}s involving $vwa$, using the 2-paths $vxc$, $vxd$, $vye$ and $vyf$. Each of these will be a valid \tm, unless the edge $ac$, $ad$, $ae$ or $af$, respectively, is present. But only two of these edges can be present, or $a$ would have degree greater than 3. Thus there is at least one valid {\tm} producing a triangle involving $v$. See Fig.~\ref{ergo:fig1}(a$'$), where the edges $ac$, $af$ are present.

  \item \emph{Exactly one pair of vertices coincide.}\ Suppose  that $d,e$ coincide and that the three vertices $c$, $f$ and $d\cdot e$ are distinct. As there are at most two edges from $a$ to these vertices, we can find at least two valid {\tm}s. See Fig.~\ref{ergo:fig1}(b) with $d,e$ identified, and edges $ac,af$ present. Then $awvx(d\cdot e)$ or $awvy(d\cdot e)$ gives a valid \tm.

  \item \emph{Exactly two pairs  of vertices coincide.}\
Suppose, without loss of generality, that  $c$ and $f$ coincide and that $d$ and $e$ coincide.  Then either of $awvx(d\cdot e)$ or $awvy(c\cdot f)$ will give
a valid make move, unless both $d\cdot e$ and $c\cdot f$ are neighbours of $a$, as in Fig.~\ref{ergo:fig1}(c). But in this case we can use vertex $b$ and
perform $\mk(bwvx(d\cdot e))$ or $\mk(bwvy(c\cdot f))$.

  \item \emph{Three pairs of vertices coincide.}\ Suppose, without loss of generality, that  $b$ and $c$ coincide, $d$ and $e$ coincide and  $a$ and $f$ coincide.  As there can be at most one extra edge incident with $a\cdot f$,
from $a\cdot f$ to
$b\cdot c$ for example, we can perform at least one of
$\mk((b\cdot c)wvx(d\cdot e))$ or $\mk((a\cdot f)wvy(d\cdot e))$.
      See Fig.~\ref{ergo:fig1}(d).

\item {\it Three vertices  all coincide.}\ Suppose  $a=c=f$, and denote this vertex as  $a\cdot c\cdot f$. Because $w,x,y$ are all neighbours of $a\cdot c\cdot f$ there can be no edge from $a\cdot c\cdot f$ to any of $b,d,e$ and we can use paths such as
$(a\cdot c\cdot f)wvye$ or $(a\cdot c\cdot f)wvxd$. 
See Fig.~\ref{ergo:fig1}(e).
\qedhere
\end{enumerate}
\end{proof}

\begin{lemma}\label{Oergo:lem20}
  Suppose that $w,x,y$ form a triangle in a component of order at least 8 in
$G\in \G_n$. Then there is a vertex $v$ such that $v,w,x,y$ induce a diamond in
a graph $G' \in \G_n$ which is  at distance at most two from $G$ in $\G_n^*$.
\end{lemma}

\begin{proof}
If $v,w,x,y$ induce a diamond for some vertex $v$ then we are done.
So we may assume that the triangle $wxyw$ is not contained in a diamond.
The situation is depicted in Fig.~\ref{ergo:fig2}(a).
Vertices $w,x,y,a,b,c$ must be distinct, since otherwise there is already a diamond or $G$ is not a simple graph. For the same reasons, there can be no edges between $w,x,y$ and $a,b,c$.
\begin{figure}[ht]
\begin{center}
\begin{tikzpicture}[scale=1]
\draw (0,0)  {} (-30:1) node[vb,label=above right:$y$] (y) {} (90:1) node[vb,label=left:$w$] (w) {} (210:1) node[vb,label=above left:$x$] (x) {};
\draw (90:2) node[vb,label=above:$a$] (a) {} (210:2) node[vb,label=left:$b$] (b) {} (-30:2) node[vb,label=right:$c\,$] (c) {} ;
\draw  (w)--(x)--(y)--(w)--(a) (x)--(b) (y)--(c) ;
\draw (0,-2) node[empty] {(a)} ;
\end{tikzpicture}
\hspace*{10mm}
\begin{tikzpicture}[scale=1]
\draw (0,0)  {} (-30:1) node[vb,label=above:$\ y$] (y) {} (90:1) node[vb,label=left:$w$] (w) {} (210:1) node[vb,label=above:$x\ $] (x) {} (90:3) node[vb,label=above:$d$] (d){};
\draw (90:2) node[vb,label=left:$a$] (a) {} (210:2) node[vb,label=left:$b$] (b) {} (-30:2) node[vb,label=right:$c\,$] (c) {} ;
\draw  (w)--(x)--(y)--(w)--(a) (x)--(b) (y)--(c) (a)--(d) (b)--(d)--(c) ;
\draw (0,-2) node[empty] {(b)} ;
\end{tikzpicture}
\\[1ex]
\begin{tikzpicture}[scale=1]
\draw (0,0)  {} (-30:1) node[vb,label=above:$\ y$] (y) {} (90:1) node[vb,label=left:$w$] (w) {} (210:1) node[vb,label=above:$x\ $] (x) {} (90:3) node[vb,label=above:$d$] (d){};
\draw (90:2) node[vb,label=left:$a$] (a) {} (210:2) node[vb,label=left:$b$] (b) {} (-30:2) node[vb,label=right:$c\,$] (c) {} (0:2.5) node[vb,label=right:$e$] (e) {} ;
\draw  (w)--(x)--(y)--(w)--(a)--(b) (x)--(b) (y)--(c) (a)--(d) (b)--(d)--(c)--(e) ;
\draw (0,-2) node[empty] {(c)} ;
\end{tikzpicture}
\hspace*{10mm}
\begin{tikzpicture}[scale=1]
\draw (0,0)  {} (-30:1) node[vb,label=above:$\ y$] (y) {} (90:1) node[vb,label=left:$w$] (w) {} (190:1) node[vb,label=above:$x\ $] (x) {} (90:3) node[vb,label=above:$d$] (d){};
\draw (90:2) node[vb,label=left:$a$] (a) {} (210:2) node[vb,label=left:$b$] (b) {} (-30:2) node[vb,label=right:$c\,$] (c) {} (45:2.5) node[vb,label=right:$e\,$] (e) {} ;
\draw  (w)--(x)--(y)--(w)--(a) (x)--(b) (y)--(c) (e)--(a)--(d) (b)--(d)--(c) (b)--(e)--(c) ;
\draw (0,-2) node[empty] {(d)} ;
\end{tikzpicture}
\end{center}
\caption{}\label{ergo:fig2}
\end{figure}
Since $G$ is 3-regular and has at least 8 vertices, at least one of $a,b,c$ must be adjacent
to a vertex $d\notin\{a,b,c,w,x,y\}$. By symmetry, we may assume that $d$ is adjacent to $a$.

Now $dawxb$ is a \tm  which creates a diamond with $v=a$, and is valid unless the edge $db$ is present. So we will assume that the edge $db$ is
present for the remainder of the proof. Similarly, the make move $dawyc$ is a \tm  which creates a diamond with $v=a$, and is valid unless the edge $dc$ is present. Again we will assume that the edge $dc$ is present, so the situation is now as
depicted in Fig.~\ref{ergo:fig2}(b).

Since $G$ is 3-regular, $a$ requires exactly one more neighbour.
We consider two cases:  in the first case, $a$ is adjacent to $b$ or $c$
(but not both, as then $a$ would have degree 4), while in the second case,
$a$ is adjacent to a vertex $e\not\in \{a,b,c,d,w,x,y\}$.

In the first case we may assume that $a$ is adjacent to $b$, by symmetry.
 Now $c$ must be adjacent to a vertex $e\notin\{a,b,c,d,w,x,y\}$, or
some vertex would have degree at least 4. For the same reason,
$e$ cannot be adjacent to any vertex in $\{a,b,d,w,x,y\}$.
But now $ecywa$ gives a valid \tm, creating a diamond with $v=c$.
See Fig.~\ref{ergo:fig2}(c).

It remains to consider the second case, namely that $a$ is adjacent
to a vertex $e\notin\{a,b,c,d,w,x,y\}$.
Now $eawyc$ is a \tm  which creates a diamond on $a,w,x,y$ and is valid unless the edge $ce$ is present. So we will assume that the edge $ce$ is present. Similarly, the make move $eawxb$ creates a diamond on $a,w,x,y$, and is valid unless the edge $eb$ is present. So again we will assume that the edge $eb$
is present, giving the situation is shown in Fig.~\ref{ergo:fig2}(d). Since
every vertex in this figure has degree 3, the subgraph shown is a connected
component of $G$ with 8 vertices.
There is no make move which will create a diamond in $G$ in one step,
so we will require two steps.

First, redraw Fig.~\ref{ergo:fig2}(d) as in Fig.~\ref{figpic}(a) below.
A make move on $bdcea$ gives Fig.~\ref{figpic}(b), and then a make move on
$bxycd$ gives Fig.~\ref{figpic}(c),which has a diamond on $w,x,c,y$.
\end{proof}

\begin{figure}[H]
\begin{center}
\begin{tikzpicture}[scale=1.25]
\draw (0.1,0.1) node[vb,label=left:$w$] (w) {} (0.9,0.1) node[vb,label=right:$x$] (x) {} (0.5,0.8) node[vb,label=left:$y$] (y) {};
\draw (x)--(y)--(w)--(x)  ;
\draw (-0.7,-0.5) node[vb,label=below:$a$] (a) {};
 \draw (1.7,-0.5) node[vb,label=below:$b$] (b) {};
\draw (0.5,1.7) node[vb,label=above:$c$] (c) {} ;
\draw (-0.9,0.9) node[vb,label=left:$d$] (d) {} ;
\draw (1.9,0.9) node[vb,label=right:$e$] (e) {} ;
\draw  (w)--(a) (x)--(b) (y)--(c) (a)--(d) (d)--(c) (b)--(e)--(c) ;
\draw (d) to[out=-80,in=180] (b);
\draw (e) to[out=-100,in=0] (a);
\draw (0.5,-1.2) node[empty] {(a)}  ;
\end{tikzpicture}
\hspace*{8mm}
\begin{tikzpicture}[scale=1.25]
\draw (0.1,0) node[vb,label=left:$w$] (w) {} (0.9,0) node[vb,label=right:$x$] (x) {} (0.5,0.7) node[vb,label=left:$y$] (y) {};
\draw (x)--(y)--(w)--(x)  ;
\draw (x)--(y)--(w)--(x)  ;
\draw (-0.7,-0.5) node[vb,label=below:$a$] (a) {};
 \draw (1.7,-0.5) node[vb,label=below:$b$] (b) {};
\draw (0.5,1.7) node[vb,label=above:$c$] (c) {} ;
\draw (-0.9,0.9) node[vb,label=left:$d$] (d) {} ;
\draw (1.9,0.9) node[vb,label=right:$e$] (e) {} ;
\draw  (w)--(a) (x)--(b) (y)--(c) (a)--(d) (d)--(c) (b)--(e)--(c) (a)--(b) ;
\draw (d) to[out=10,in=170] (e);
\draw (0.5,-1.2) node[empty] {(b)}  ;
\end{tikzpicture}\hspace*{8mm}
\begin{tikzpicture}[scale=1.25]
\draw (0.1,0) node[vb,label=above :$w\ \,$] (w) {} (0.9,0) node[vb,label=right:$x$] (x) {} (0.5,0.7) node[vb,label=left:$y$] (y) {};
\draw (x)--(y)--(w)--(x)  ;
\draw (x)--(y)--(w)--(x)  ;
\draw (-0.7,-0.5) node[vb,label=below:$a$] (a) {};
 \draw (1.7,-0.5) node[vb,label=below:$b$] (b) {};
\draw (0.75,1.7) node[vb,label=above:$c$] (c) {} ;
\draw (-0.9,0.9) node[vb,label=left:$d$] (d) {} ;
\draw (1.9,0.9) node[vb,label=right:$e$] (e) {} ;
\draw  (w)--(a) (y)--(c) (a)--(d) (x)--(c) (b)--(e)--(c) (a)--(b) ;
\draw (d) to[out=10,in=170] (e);
\draw (d) to[out=-70,in=170] (b);
\draw (0.5,-1.2) node[empty] {(c)}  ;
\end{tikzpicture}
\end{center}
\caption{}\label{figpic}
\end{figure}

Now we have reduced the connectivity question to graphs containing at least one diamond. We show that any diamond can be transformed to a $K_4$ in at most
two steps.

\begin{lemma}\label{ergo:lem30}
If vertices $v,w,x,y$ span a diamond in $G\in \G_n$ then $G$ is connected in $\G_n$ to a graph $G'$ such that $v,w,x,y$ induce a $K_4$ component. The path from $G$ to $G'$ in $\G_n^*$
has length at most two.
\end{lemma}

\begin{proof}
The situation is depicted in Fig.~\ref{ergo:fig3}. Vertices $v,y$ must be adjacent to some vertex not in the diamond, since otherwise either $\{v,w,x,y\}$ is already a $K_4$ or some vertex does not have degree 3. There are two cases, shown in Fig.~\ref{ergo:fig3}. In case (a),
the final neighbour of $v$ is $a$ and the final neighbour of $y$ is $b$, where $a\neq b$.
Whether or not $ab$ is an edge, $a$ must have another neighbour $c\not\in\{v,w,x,y,b\}$.
For the same reason, $c$ must be adjacent to some vertex $d\notin\{v,w,x,y,a,b,c\}$
(regardless of whether the edge $bc$ is present).
So the situation is as depicted in Fig.~\ref{ergo:fig3}(a). In case (b), both $v$ and $y$ are
adjacent to the same vertex $a$. Now $a$ must be adjacent to some vertex $b\notin\{v,w,x,y\}$,
by 3-regularity. Similarly, $b$ must be adjacent to some vertex $c\notin\{v,w,x,y,a,b\}$, and $c$ must be adjacent to some vertex $d\notin\{v,w,x,y,a,b,c\}$. So the situation is as depicted in Fig.~\ref{ergo:fig3}(b).
\begin{figure}[H]
\begin{center}
\begin{tikzpicture}[xscale=0.9]
\draw  (0.75,1) node[vb,label=above:$v$] (v) {} (0,0) node[vb,label=left:$w$] (w) {} (1.5,0) node[vb,label=right:$x$] (x) {} (0.75,-1) node[vb,label=below:$y$] (y) {};
\draw (2.5,1) node[vb,label=above:$a$] (a) {} (2.5,-1) node[vb,label=below:$b$] (b) {} (4,1) node[vb,label=above:$c$] (c) {} (5.5,1) node[vb,label=above:$d$] (d) {} ;
\draw (w)--(x)--(y)--(w)--(v)--(x) (v)--(a)--(c)--(d) (y)--(b) ;
\draw (2.75,-2) node[empty] {(a)} ;
\end{tikzpicture}
\hspace*{20mm}
\begin{tikzpicture}[xscale=0.9]
\draw  (0.75,1) node[vb,label=above:$v$] (v) {} (0,0) node[vb,label=left:$w$] (w) {} (1.5,0) node[vb,label=right:$x$] (x) {} (0.75,-1) node[vb,label=below:$y$] (y) {};
\draw (3,0) node[vb,label=above:$a$] (a) {} (4.5,0) node[vb,label=above:$b$] (b) {} (6,0) node[vb,label=above:$c$] (c) {} (7.5,0) node[vb,label=above:$d$] (d) {} ;
\draw (w)--(x)--(y)--(w)--(v)--(x) (v)--(a) (y)--(a)--(b)--(c)--(d) ;
\draw (3.75,-2) node[empty] {(b)} ;
\end{tikzpicture}
\end{center}
\caption{}\label{ergo:fig3}
\end{figure}
First consider case (a). The operation $\mk(avwyb)$ will insert edges
$vy$, $ab$ and delete edges $av$, $by$. Then $v,w,x,y$ will induce a $K_4$ component. The \tm  will be valid provided
that $ab\notin E$. So we will assume that the edge $ab$ is present. See Fig.~\ref{ergo:fig4}(i). Now a \tm  using $dcaby$ will insert edges $bc$, $dy$ and delete edges $cd$, $by$. This move will be valid unless $bc\in E$, since $dy\notin E$ or $y$ could not have degree 3. If $bc\notin E$ then the \tm  will give the graph in Fig.~\ref{ergo:fig4}(ii). Now the \tm  using $avwyd$ will insert $vy$, $ad$ and delete $av$, $dy$. This is valid, since $ad\notin E$ or else $a$ could not have degree 3. After these operations, vertices $v,w,x,y$ induce a $K_4$ component.
\begin{figure}[H]
\begin{center}
\begin{tikzpicture}[scale=1]
\draw  (0.75,1) node[vb,label=above:$v$] (v) {} (0,0) node[vb,label=left:$w$] (w) {} (1.5,0) node[vb,label=right:$x$] (x) {} (0.75,-1) node[vb,label=below:$y$] (y) {};
\draw (2.5,1) node[vb,label=above:$a$] (a) {} (2.5,-1) node[vb,label=below:$b$] (b) {} (3.5,1) node[vb,label=above:$c$] (c) {} (4.5,1) node[vb,label=above:$d$] (d) {} ;
\draw (w)--(x)--(y)--(w)--(v)--(x) (v)--(a) (y)--(b)--(a)--(c)--(d) ;
\draw (1.25,-2) node[empty] {(i)} ;
\end{tikzpicture}
\hspace*{20mm}
\begin{tikzpicture}[scale=1]
\draw  (0.75,1) node[vb,label=above:$v$] (v) {} (0,0) node[vb,label=left:$w$] (w) {} (1.5,0) node[vb,label=right:$x$] (x) {} (0.75,-1) node[vb,label=below:$y$] (y) {};
\draw (2.5,1) node[vb,label=above:$a$] (a) {} (2.5,-1) node[vb,label=below:$b$] (b) {} (3.5,1) node[vb,label=above:$c$] (c) {} (4.5,1) node[vb,label=above:$d$] (d) {} ;
\draw (w)--(x)--(y)--(w)--(v)--(x) (v)--(a) (b)--(a)--(c)--(b) (d)--(y) ;
\draw (1.25,-2) node[empty] {(ii)} ;
\end{tikzpicture}
\end{center}
\caption{}\label{ergo:fig4}
\end{figure}
Now suppose that $bc\in E$, as in Fig.~\ref{ergo:fig5}(i). Then $d$ must be adjacent to some vertex $e\notin\{v,w,x,y,a,b,c\}$, or $G$ could not be 3-regular. The \tm  using $edcby$ will insert $bd$, $ey$ and delete $de$, $by$. This move is valid, since 3-regularity implies that $bd,ey\notin E$,
and produces the graph shown in Fig.~\ref{ergo:fig5}(ii).
\begin{figure}[H]
\begin{center}
\begin{tikzpicture}[xscale=0.9]
\draw  (0.75,1) node[vb,label=above:$v$] (v) {} (0,0) node[vb,label=left:$w$] (w) {} (1.5,0) node[vb,label=right:$x$] (x) {} (0.75,-1) node[vb,label=below:$y$] (y) {};
\draw (2.5,1) node[vb,label=above:$a$] (a) {} (2.5,-1) node[vb,label=below:$b$] (b) {} (4,1) node[vb,label=above:$c$] (c) {} (5.5,1) node[vb,label=above:$d$] (d) {} (7,1) node[vb,label=above:$e$] (e) {} ;
\draw (w)--(x)--(y)--(w)--(v)--(x) (v)--(a)--(c)--(d)--(e) (y)--(b) (a)--(b)--(c);
\draw (2.75,-2) node[empty] {(i)} ;
\end{tikzpicture}
\hspace*{15mm}
\begin{tikzpicture}[xscale=0.9]
\draw  (0.75,1) node[vb,label=above:$v$] (v) {} (0,0) node[vb,label=left:$w$] (w) {} (1.5,0) node[vb,label=right:$x$] (x) {} (0.75,-1) node[vb,label=below:$y$] (y) {};
\draw (2.5,1) node[vb,label=above:$a$] (a) {} (2.5,-1) node[vb,label=below:$b$] (b) {} (4,1) node[vb,label=above:$c$] (c) {} (5.5,1) node[vb,label=above:$d$] (d) {} (7,1) node[vb,label=above:$e$] (e) {} ;
\draw (w)--(x)--(y)--(w)--(v)--(x) (v)--(a)--(c)--(d) (y)--(e) (b)--(d) (a)--(b)--(c)
;
\draw (2.75,-2) node[empty] {(ii)} ;
\end{tikzpicture}
\end{center}
\caption{}\label{ergo:fig5}
\end{figure}
Now the \tm  using $avwye$ will insert $vy,ae$ and delete $va,ey$.  This move will be valid since $ae\notin E$,or $a$ would not have degree 3. Then $v,w,x,y$ will again induce a $K_4$ component.

Finally, consider case (b), shown in Fig.~\ref{ergo:fig6}(i). The \tm  using $dcbav$ will insert $ac,dv$, and delete $cd,av$. This move is valid, since $ac,ey\notin E$ by 3-regularity. The \tm  gives the graph in Fig.~\ref{ergo:fig6}(ii).
\begin{figure}[H]
\begin{center}
\begin{tikzpicture}[xscale=0.9]
\draw  (0.75,1) node[vb,label=above:$v$] (v) {} (0,0) node[vb,label=left:$w$] (w) {} (1.5,0) node[vb,label=right:$x$] (x) {} (0.75,-1) node[vb,label=below:$y$] (y) {};
\draw (3,0) node[vb,label=above:$a$] (a) {} (4.5,0) node[vb,label=above:$b$] (b) {} (6,0) node[vb,label=above:$c$] (c) {} (7.5,0) node[vb,label=above:$d$] (d) {} ;
\draw (w)--(x)--(y)--(w)--(v)--(x) (v)--(a) (y)--(a)--(b)--(c)--(d) ;
\draw (3.75,-2) node[empty] {(i)} ;
\end{tikzpicture}
\hspace*{10mm}
\begin{tikzpicture}[xscale=0.9]
\draw  (0.75,1) node[vb,label=above:$v$] (v) {} (0,0) node[vb,label=left:$w$] (w) {} (1.5,0) node[vb,label=right:$x$] (x) {} (0.75,-1) node[vb,label=below:$y$] (y) {};
\draw (3,-0.25) node[vb,label=above:$a$] (a) {} (4.5,-0.25) node[vb,label=above:$b$] (b) {} (6,-0.25) node[vb,label=above:$c$] (c) {} (7.5,0) node[vb,label=above:$d$] (d) {} ;
\draw (w)--(x)--(y)--(w)--(v)--(x) (a)--(y) (a)--(b)--(c) (d)--(v)(a)edge[bend right=20](c);
\draw (3.75,-2) node[empty] {(ii)} ;
\end{tikzpicture}
\end{center}
\caption{}\label{ergo:fig6}
\end{figure}
 Now the \tm  using $aywvd$ will insert $vy,ad$ and delete $ay,dv$.  This move will be valid since $ad\notin E$,or $a$ would not have degree 3. Then $v,w,x,y$ will again induce a $K_4$ component.
\end{proof}

We can now complete the proof of Theorem~\ref{ergo:thm1}.

\begin{proof}[Proof of Theorem~\ref{ergo:thm1}]
Suppose that $n\equiv 0\! \pmod{4}$.
The procedure implied by Lemmas~\ref{ergo:lem10}--\ref{ergo:lem30} shows that, for any $G$ with a component of order at least 8, there is a path of length at most 4 in $\G_n$ from $G$ to a graph $G'$ having a component $C$ isomorphic to $K_4$. Applying this inductively to $G\setminus C$ must result in a graph  $H$ which has components only of order 4 or 6. Suppose there are two components of order 6, as shown in Fig.~\ref{2x6}(a) below. Perform the break move $\bk(abc,de)$, leading to Fig.~\ref{2x6}(b).
Now we can perform the make move $\mk(dbvxc)$ to create a diamond on vertices $b,v,w,x$. Note that we have created a diamond in two moves, in agreement with Lemma~\ref{Oergo:lem20}. The remaining 8 vertices form a component of the type shown in Fig.~\ref{figpic}(b). A further make move gives a component with a diamond, as shown in Fig.~\ref{figpic}(c). We can then use Lemma~\ref{ergo:lem30} to create a $K_4$.

The final configuration must consist of $(n/4-1)$ components isomorphic to $K_4$ and a residual component of order $4$. But $K_4$ is the only 3-regular labelled graph on 4 vertices, so $H$ consists of $n/4$ components, each of which is a $K_4$.
\end{proof}

\begin{figure}[ht]
\begin{center}
\begin{tikzpicture}[scale=1]
\draw (0.1,0) node[vb,label=below:$w$] (0) {} (0.9,0) node[vb,label=below:$x$] (1) {} (0.5,0.7) node[vb,label=left:$v$] (2) {}
(-0.7,-0.5) node[vb,label=below:$a$] (3) {} (1.7,-0.5) node[vb,label=below:$c$] (4) {} (0.5,1.7) node[vb,label=left:$b$] (5) {} ;
\draw (0)--(1)--(2)--(0) (3)--(4)--(5)--(3) (0)--(3) (1)--(4) (2)--(5) ;
\begin{scope}[xshift=4cm]
\draw (0.1,0) node[vb] (0) {} (0.9,0) node[vb] (1) {} (0.5,0.7) node[vb] (2) {}
(-0.7,-0.5) node[vb,label=below:$e$] (3) {} (1.7,-0.5) node[vb] (4) {} (0.5,1.7) node[vb,label=right:$d$] (5) {} ;
\draw (0)--(1)--(2)--(0) (3)--(4)--(5)--(3) (0)--(3) (1)--(4) (2)--(5) ;
\end{scope}
\draw (2.5,-1.25) node[empty] {(a)} ;
\end{tikzpicture}
\hspace*{2cm}
\begin{tikzpicture}[scale=1]
\draw (0.1,0) node[vb,label=below:$w$] (0) {} (0.9,0) node[vb,label=below:$x$] (1) {} (0.5,0.7) node[vb,label=left:$v$] (2) {}
(-0.7,-0.5) node[vb,label=below:$a$] (3) {} (1.7,-0.5) node[vb,label=below:$c$] (c) {} (0.5,1.7) node[vb,label=left:$b$] (b) {} ;
\draw (0)--(1)--(2)--(0) (3)--(c) (b)--(3) (0)--(3) (1)--(c) (2)--(b) ;
\begin{scope}[xshift=4cm]
\draw (0.1,0) node[vb] (0) {} (0.9,0) node[vb] (1) {} (0.5,0.7) node[vb] (2) {}
(-0.7,-0.5) node[vb,label=below:$e$] (3) {} (1.7,-0.5) node[vb] (4) {} (0.5,1.7) node[vb,label=right:$d$] (5) {} ;
\draw (0)--(1)--(2)--(0) (3)--(4)--(5) (0)--(3) (1)--(4) (2)--(5) (b)--(5) (c)--(3) ;
\end{scope}
\draw (2.5,-1.25) node[empty] {(b)} ;
\end{tikzpicture}
\end{center}
\caption{}\label{2x6}\vspace{2ex}
\end{figure}

\subsection*{Proof of  Theorem~\ref{TH1}}

In order to establish that any triangle switch chain is irreducible,
it remains to prove that
any two graphs in $\K^*_n$ are connected by a path in $\G^*_n$.
\begin{theorem}\label{ergo:thm2}
 If $n\equiv 0\!\pmod 4$ then any two graphs in $\K^*_n$ are connected by a path in $\G_n^*$.
\end{theorem}
\begin{proof}
Suppose that $H,H'$ are any two graphs in $\K_n$. Let $H$ have components $C_1,C_2,\ldots,C_{n/4}$ and  $H'$ have components $C'_1,C'_2,\ldots,C'_{n/4}$.
If $C_1$ and $C_j$ are two components of $H$ then we can exchange any pair of vertices
$v\in C_1$, $a\in C_j$ using the following moves.

The  \tb $\bk(ywv,ab)$ deletes the edges $vw,ab$, and inserts the edges $wa,vb$.
See Fig.~\ref{ergo:fig7}(ii). Next, the make move $\mk(vywad)$ deletes the edges $vy,ad$ and inserts the edges $ay,vd$.
The resulting graph is shown in Fig.~\ref{ergo:fig7}(iii).
Finally, the make move $\mk(vxyac)$ deletes the edges $vx$, $ac$ and
inserts the edges $vc$, $ax$, leading to the graph shown
in Fig.~\ref{ergo:fig7}(iv).
Observe that $v$ and $a$ have been exchanged between the two $K_4$'s, and that
by symmetry, any chosen pair of vertices could be exchanged in this way.

\begin{figure}[ht]
\begin{center}
\begin{tikzpicture}[scale=1]
\draw (0,0) node[vb,label=below:$y$] (y) {} (-30:1) node[vb,label=below:$v$] (v) {} (90:1) node[vb,label=above:$w$] (w) {} (210:1) node[vb,label=below:$x$] (x) {};
\draw  (v)--(w)--(x)--(y)--(v)--(x) (y)--(w);
\begin{scope}[xshift=3cm]
\draw (0,0) node[vb,label=below:$d$] (d) {} (-30:1) node[vb,label=below:$c$] (c) {} (90:1) node[vb,label=above:$a$] (a) {} (210:1) node[vb,label=below:$b$] (b) {};
\draw  (c)--(b)--(d)--(c)--(a)--(b) (a)--(d);
\end{scope}
\draw (1.5,-1.5) node[empty] {(i)} ;
\begin{scope}[xshift=8cm]
\draw (0,0) node[vb,label=below:$y$] (y) {} (-30:1) node[vb,label=below:$v$] (v) {} (90:1) node[vb,label=above:$w$] (w) {} (210:1) node[vb,label=below:$x$] (x) {};
\begin{scope}[xshift=3cm]
\draw (0,0) node[vb,label=below:$d$] (d) {} (-30:1) node[vb,label=below:$c$] (c) {} (90:1) node[vb,label=above:$a$] (a) {} (210:1) node[vb,label=below:$b$] (b) {};
\end{scope}
\draw  (v)--(x)--(w)--(y)--(v)--(b)--(d)--(a)--(c)--(b) (y)--(x) (c)--(d) (a)--(w);
\draw (1.5,-1.5) node[empty] {(ii)} ;
\end{scope}
\end{tikzpicture}\\[2ex]
\begin{tikzpicture}[scale=1]
\draw (0,1) node[vb,label=left:$x$] (x) {};
\draw (2,0) node[vb,label=below:$v$] (v) {};
\draw (3.5,0) node[vb,label=below:$b$] (b) {};
\draw (4,1) node[vb,label=right:$c$] (c) {};
\draw (2.5,1) node[vb,label=above:$d$] (d) {};
\draw (2,2) node[vb,label=above:$a$] (a) {};
\draw (0.5,2) node[vb,label=above:$w$] (w) {};
\draw (1.5,1) node[vb,label=below:$y$] (y) {};
\draw  (y)--(x)--(w)--(a)--(y)--(w) (a)--(c)--(b)--(v)--(d)--(b) (d)--(c) (x)--(v);
\draw (1.5,-1.5) node[empty] {(iii)} ;
\begin{scope}[xshift=8cm]
\draw (0,1) node[vb,label=left:$x$] (x) {};
\draw (2,0) node[vb,label=below:$v$] (v) {};
\draw (3.5,0) node[vb,label=below:$b$] (b) {};
\draw (4,1) node[vb,label=right:$c$] (c) {};
\draw (2.5,1) node[vb,label=above:$d$] (d) {};
\draw (2,2) node[vb,label=above:$a$] (a) {};
\draw (0.5,2) node[vb,label=above:$w$] (w) {};
\draw (1.5,1) node[vb,label=below:$y$] (y) {};
\draw  (y)--(x)--(w)--(a)--(y)--(w) (a)--(x)  (b)--(v)--(d)--(b)--(c)--(d) (c)--(v);
\draw (1.5,-1.5) node[empty] {(iv)} ;
\end{scope}
\end{tikzpicture}
\end{center}
\caption{}\label{ergo:fig7}
\end{figure}
We can use this procedure to transform $H$ to $H''\in\K_n$ so that its first component is $C'_1$. We simply locate the (at most four) components of $H$ which contain the vertices in $C'_1$, and switch these vertices into $C_1$. Then we apply the same argument inductively to the graphs $H''\setminus C'_1$ and $H'\setminus C'_1$. This gives a recursive construction for a
path from $H$ to $H'$ in $\G^*_n$.
\end{proof}
This gives the required result.
\begin{theorem}\label{ergo:thm3}
  If $n\equiv 0\!\pmod 4$ then $\G^*_n$ is connected.
\end{theorem}
\begin{proof}
This follows directly from Theorems~\ref{ergo:thm1} and~\ref{ergo:thm2}.
\end{proof}
We must now consider the case $n\equiv 2\!\pmod 4$. Necessarily, $n\geq 6$, or there are no 3-regular simple graphs with $n$ vertices. We redefine $\K_n$ to be be the class of labelled graphs comprising $(n-6)/4$ components isomorphic to $K_4$ and one component of order 6.
Thus we must consider $\G^*_6$. Unlike $\G^*_4$, this is not a single graph. In fact, it contains two non-isomorphic graphs, the complete bipartite graph \ktt, and the triangular prism, \tpr. The graph \tpr is the complement of a 6-cycle, hence the notation.
\begin{figure}[H]
\begin{center}
\begin{tikzpicture}[xscale=1.5,yscale=1.7]
\draw (0,0) node[vb] (0) {} (1,0) node[vb] (1) {} (2,0) node[vb] (2) {}
(0,1) node[vb] (3) {} (1,1) node[vb] (4) {} (2,1) node[vb] (5) {} ;
\draw (0)--(3)(0)--(4)(0)--(5) (1)--(3)(1)--(4)(1)--(5) (2)--(3)(2)--(4)(2)--(5) ;
\draw (1,-0.7) node[empty] {\ktt} ;
\end{tikzpicture}
\hspace*{3cm}
\begin{tikzpicture}[scale=1]
\draw (0.1,0) node[vb] (0) {} (0.9,0) node[vb] (1) {} (0.5,0.7) node[vb] (2) {}
(-0.7,-0.5) node[vb] (3) {} (1.7,-0.5) node[vb] (4) {} (0.5,1.7) node[vb] (5) {} ;
\draw (0)--(1)--(2)--(0) (3)--(4)--(5)--(3) (0)--(3) (1)--(4) (2)--(5) ;
\draw (1,-1.7) node[empty] {\tpr} ;
\end{tikzpicture}
\end{center}
\caption{}\label{ergo:fig8}
\end{figure}

In fact, there are  10 different labellings for \ktt and 60 for \tpr, so $|\G^*_6|=70$, though we will not need these numbers. But we need to show that $\G^*_6$ is connected. Let $\A$ be the set of labelled graphs in $\G^*_6$ which are isomorphic to \ktt, and let $\B$ be the set which are isomorphic to \tpr. First we show
\begin{lemma}\label{ergo:lem40}
  Every graph in \B is adjacent in $\G^*_6$ to a graph in \A.
\end{lemma}
\begin{proof}
The graph in Fig.~\ref{ergo:fig9}(i) belongs to $\B$. In this graph, perform
the \tb $\bk(fba,de)$,
which deletes the edges $ab,de$ and inserts the edges $bd,ae$. The resulting graph is shown in Fig.~\ref{ergo:fig9}(ii), and redrawn as \ktt in Fig.~\ref{ergo:fig9}(iii).\qedhere
\begin{figure}[H]
\begin{center}
\begin{tikzpicture}[scale=1.25]
\draw (0.1,0) node[vb,label=below:$a$] (0) {} (0.9,0) node[vb,label=below:$b$] (1) {} (0.5,0.7) node[vb,label=left:$f$] (2) {}
(-0.7,-0.5) node[vb,label=left:$d$] (3) {} (1.7,-0.5) node[vb,label=right:$e$] (4) {} (0.5,1.7) node[vb,label=left:$c$] (5) {} ;
\draw (0)--(1)--(2)--(0) (3)--(4)--(5)--(3) (0)--(3) (1)--(4) (2)--(5) ;
\draw (0.5,-1) node[empty] {(i)} ;
\end{tikzpicture}
\hspace*{1.5cm}
\begin{tikzpicture}[scale=1.25]
\draw (0.1,0) node[vb,label=above:$a\ $] (0) {} (0.9,0) node[vb,label=above:$\ b$] (1) {} (0.5,0.7) node[vb,label=left:$f$] (2) {}
(-0.7,-0.5) node[vb,label=left:$d$] (3) {} (1.7,-0.5) node[vb,label=right:$e$] (4) {} (0.5,1.7) node[vb,label=left:$c$] (5) {} ;
\draw (1)--(2)--(0) (4)--(5)--(3) (0)--(3) (1)--(4) (2)--(5) (3)--(1) (0)--(4) ;
\draw (0.5,-1) node[empty] {(ii)} ;
\end{tikzpicture}
\hspace*{1.5cm}
\begin{tikzpicture}[xscale=1.5,yscale=2]
\draw (0,0) node[vb,label=below:$a$] (0) {} (1,0) node[vb,label=below:$b$] (1) {} (2,0) node[vb,label=below:$c$] (2) {}
(0,1) node[vb,label=above:$d$] (3) {} (1,1) node[vb,label=above:$e$] (4) {} (2,1) node[vb,label=above:$f$] (5) {} ;
\draw (0)--(3)(0)--(4)(0)--(5) (1)--(3)(1)--(4)(1)--(5) (2)--(3)(2)--(4)(2)--(5) ;
\draw (1,-0.5) node[empty] {(iii)} ;
\end{tikzpicture}
\end{center}
\caption{}\label{ergo:fig9}
\end{figure}
\end{proof}
Next we will show that any two graphs in $\A$ are connected in  $\G^*_6$.
\begin{lemma}\label{ergo:lem50}
  All graphs in \A are connected by a path in $\G^*_6$.
\end{lemma}
\begin{proof}
Note that $G\in\A$ is completely determined by the unordered partition $V_1,V_2$ of its vertex set $V$ such that $|V_1|=|V_2|=3$. For example, in Fig.~\ref{ergo:fig9}(iii),
$V_1=\{a,b,c\},\,V_2=\{d,e,f\}$, or vice versa. Now the lemma will follow from showing that we can switch any vertex in $V_1$ with any vertex in $V_2$. Consider the graph in Fig.~\ref{ergo:fig9}(i), reproduced in Fig.~\ref{ergo:fig10}(i). This results from a \tm applied to the graph in Fig.~\ref{ergo:fig9}(iii).
\begin{figure}[H]
\begin{center}
\begin{tikzpicture}[scale=1.25]
\draw (0.1,0) node[vb,label=below:$a$] (0) {} (0.9,0) node[vb,label=below:$b$] (1) {} (0.5,0.7) node[vb,label=left:$f$] (2) {}
(-0.7,-0.5) node[vb,label=left:$d$] (3) {} (1.7,-0.5) node[vb,label=right:$e$] (4) {} (0.5,1.7) node[vb,label=left:$c$] (5) {} ;
\draw (0)--(1)--(2)--(0) (3)--(4)--(5)--(3) (0)--(3) (1)--(4) (2)--(5) ;
\draw (0.5,-1) node[empty] {(i)} ;
\end{tikzpicture}
\hspace*{1.5cm}
\begin{tikzpicture}[scale=1.25]
\draw (0.1,0) node[vb,label=below:$a$] (0) {} (0.9,0) node[vb,label=below:$b$] (1) {} (0.5,0.7) node[vb,label=left:$f$] (2) {}
(-0.7,-0.5) node[vb,label=left:$d$] (3) {} (1.7,-0.5) node[vb,label=right:$e$] (4) {} (0.5,1.7) node[vb,label=left:$c$] (5) {} ;
\draw (2)--(0)--(1) (3)--(4) (5)--(3) (0)--(3) (1)--(4) (4)--(2)--(5)--(1) ;
\draw (0.5,-1) node[empty] {(i)} ;
\end{tikzpicture}
\hspace*{1.5cm}
\begin{tikzpicture}[xscale=1.5,yscale=2]
\draw (0,0) node[vb,label=below:$a$] (0) {} (1,0) node[vb,label=below:$e$] (1) {} (2,0) node[vb,label=below:$c$] (2) {}
(0,1) node[vb,label=above:$d$] (3) {} (1,1) node[vb,label=above:$b$] (4) {} (2,1) node[vb,label=above:$f$] (5) {} ;
\draw (0)--(3)(0)--(4)(0)--(5) (1)--(3)(1)--(4)(1)--(5) (2)--(3)(2)--(4)(2)--(5) ;
\draw (1,-0.5) node[empty] {(iii)} ;
\end{tikzpicture}
\end{center}
\caption{}\label{ergo:fig10}
\end{figure}
Suppose we apply the \tb $\bk(abf,ce)$. This deletes the edges $bf,ce$ and inserts the edges $bc,ef$, giving the graph in Fig.~\ref{ergo:fig10}(ii), redrawn in Fig.~\ref{ergo:fig10}(iii). Now $b,e$ have swapped between $V_1$ and $V_2$ from Fig.~\ref{ergo:fig9}(iii). Similarly, using the symmetries of \tpr, we can apply a \tm followed by a \tb in this way to switch any pair of vertices between $V_1$ and $V_2$. Then two such switches will transform any graph in \A into any other graph in \A.
\end{proof}
Thus we have:
\begin{lemma}\label{ergo:lem60}
  $\G^*_6$ is connected.
\end{lemma}
\begin{proof}
This follows directly from Lemmas~\ref{ergo:lem40} and~\ref{ergo:lem50}.
\end{proof}
We can now prove the general result.
\begin{theorem}\label{ergo:thm4}
  $\G^*_n$ is a connected graph for any $n$.
\end{theorem}
\begin{proof}
We may assume that $n\equiv 2\!\pmod4$, and $n\geq 8$, since otherwise the
result follows from Theorem~\ref{ergo:thm3} and Lemma~\ref{ergo:lem60}. Now the proof of Theorem~\ref{ergo:thm1} remains valid when $n\equiv 2\pmod{4}$,
except that the final graph $H$ will now comprise $(n-6)/4$ $K_4$'s and one component of order 6. By Lemma~\ref{ergo:lem60}, we may assume that this component is a \tpr. The proof of Theorem~\ref{ergo:thm2} will also hold true, provided that we can switch vertices between a $K_4$ and a \tpr, as we did between two $K_4$'s in Theorem~\ref{ergo:thm2}. We complete the proof by showing that this can be done in a similar way.
\begin{figure}[ht]
\begin{center}
\begin{tikzpicture}[scale=1.25]
\draw (0,0) node[vb,label=below:$y$] (y) {} (-30:1) node[vb,label=below:$v$] (v) {} (90:1) node[vb,label=above:$w$] (w) {} (210:1) node[vb,label=below:$x$] (x) {};
\draw  (v)--(w)--(x)--(y)--(v)--(x) (y)--(w);
\begin{scope}[scale=0.9,xshift=2.5cm]
\draw (0.1,0) node[vb,label=below:$a$] (0) {} (0.9,0) node[vb,label=below:$b$] (1) {} (0.5,0.7) node[vb,label=left:$f$] (2) {}
(-0.7,-0.5) node[vb,label=below:$d$] (3) {} (1.7,-0.5) node[vb,label=below:$e$] (4) {} (0.5,1.7) node[vb,label=left:$c$] (5) {} ;
\draw (0)--(1)--(2)--(0) (3)--(4)--(5)--(3) (0)--(3) (1)--(4) (2)--(5) ;
\end{scope}
\draw (1.5,-1.25) node[empty] {(i)} ;
\end{tikzpicture}
\hspace*{2cm}
\begin{tikzpicture}[scale=1.25]
\draw (0,0) node[vb,label= right:$\ y$] (y) {} (-30:1) node[vb,label=below:$v$] (v) {} (90:1) node[vb,label=above:$w$] (w) {} (190:1) node[vb,label=below:$x$] (x) {};
\draw  (w)--(x)--(y)--(v)--(x) (y)--(w);
\begin{scope}[scale=0.9,xshift=2.5cm]
\draw (0.1,0) node[vb,label=below:$a$] (0) {} (0.9,0) node[vb,label=below:$b$] (1) {} (0.5,0.7) node[vb,label=left:$f$] (2) {}
(-0.7,-0.5) node[vb,label=below:$d$] (3) {} (1.7,-0.5) node[vb,label=below:$e$] (4) {} (0.5,1.7) node[vb,label=right:$c$] (5) {} ;
\draw (0)--(1)--(2)--(0) (3)--(4)--(5) (0)--(3) (1)--(4) (2)--(5) (w)--(5) (v)--(3);
\end{scope}
\draw (1.5,-1.25) node[empty] {(ii)} ;
\end{tikzpicture}\\[1ex]
\begin{tikzpicture}[scale=1.25]
\draw (0,0) node[vb,label=below:$y$] (y) {} (-30:1) node[vb,label=below:$v$] (v) {} (90:1) node[vb,label=above:$w$] (w) {} (190:1) node[vb,label=below:$x$] (x) {};
\draw  (w)--(x)--(y)--(v) (y)--(w);
\begin{scope}[scale=0.9,xshift=2.5cm]
\draw (0.1,0) node[vb,label=below:$a$] (0) {} (0.9,0) node[vb,label=below:$b$] (1) {} (0.5,0.7) node[vb,label=above:$f$] (2) {}
(-0.7,-0.5) node[vb,label=below:$d$] (3) {} (1.7,-0.5) node[vb,label=below:$e$] (4) {} (0.5,1.7) node[vb,label=right:$c$] (5) {} ;
\draw (0)--(1)--(2)--(0) (3)--(4)--(5) (0)--(3) (1)--(4) (w)--(5) (v)--(3) (x)--(5) (v)--(2);
\end{scope}
\draw (1.5,-1.25) node[empty] {(iii)} ;
\end{tikzpicture}
\hspace*{3cm}
\begin{tikzpicture}[scale=1.25]
\draw (0,0) node[vb,label=below:$y$] (y) {} (-30:1) node[vb,label=below:$c$] (v) {} (90:1) node[vb,label=above:$w$] (w) {} (210:1) node[vb,label=below:$x$] (x) {};
\draw  (v)--(w)--(x)--(y)--(v)--(x) (y)--(w);
\begin{scope}[scale=0.9,xshift=2.5cm]
\draw (0.1,0) node[vb,label=below:$a$] (0) {} (0.9,0) node[vb,label=below:$b$] (1) {} (0.5,0.7) node[vb,label=left:$f$] (2) {}
(-0.7,-0.5) node[vb,label=below:$d$] (3) {} (1.7,-0.5) node[vb,label=below:$e$] (4) {} (0.5,1.7) node[vb,label=left:$v$] (5) {} ;
\draw (0)--(1)--(2)--(0) (3)--(4)--(5)--(3) (0)--(3) (1)--(4) (2)--(5) ;
\end{scope}
\draw (1.5,-1.25) node[empty] {(iv)} ;
\end{tikzpicture}
\end{center}
\caption{}\label{ergo:fig11}\vspace{2ex}
\end{figure}

In Fig.~\ref{ergo:fig11}(i), apply the  \tb $(ywv,cd)$, deleting edges $vw,cd$ and inserting $cw,vd$, giving the graph in Fig.~\ref{ergo:fig11}(ii). Now apply the \tm $vxwcf$, deleting edges $vx,cf$ and inserting $cx,vf$, giving the graph in igure~\ref{ergo:fig11}(iii).
Finally, apply the \tm $vywce$, deleting edges
$vy,ce$ and inserting $cy,ve$, giving the graph in Fig.~\ref{ergo:fig11}(iii). We have exchanged $v$ and $c$ between the $K_4$ and \tpr. By symmetry, we can switch any pair in this way.
\end{proof}
This completes the proof of Theorem~\ref{TH1}.

\newpage

\section{Formal definitions of Chain~I and Chain~II}\label{proc1:sec}

Before proving our remaining results, we must define Chains~I and~II more formally.
Recall the make and break moves shown schematically in Fig.~\ref{proc:fig1}.
The procedure $\bk$ replaces a triangle and a disjoint~edge with a path of length~4.
\vspace{-0.05in}
\begin{tabbing}
aaa \= aa \= aa \= aa \= aa \kill
\tabrule\\
$\bk(vxw,yz)$\\
\>\>$E\gets (E\setminus\{xw,yz\})\cup\{xy,wz\}$\\
{\bf end}\\
\tabrule\\
\end{tabbing}
\vspace{-0.05in}
The procedure $\mk$ replaces a path $yxvwz$ of length 4 with a triangle $vxw$ and a disjoint edge $yz$.
\begin{tabbing}
aaa \= aa \= aa \= aa \= aa \kill
\tabrule\\
$\mk(yxvwz)$\\[0.5ex]
\>\>$E\gets (E\setminus\{xy,wz\})\cup\{xw,yz\}$\\
{\bf end}\\
\tabrule
\end{tabbing}
Note that there is always a \tb which reverses a given \tm,
and vice-versa.

\ignore{
Note that there is a unique \tb which reverses any \tm, and vice versa.
\csg{Not really unique:  there's the mirror images, which we treat as the same.  But more seriously,
there can be 2, 3 or 4 different paths of length 4 which give the same transition $G\mapsto G'$, as shown
in Fig.~\ref{fig:metro}.
}}

Recall that $\G_n$ is the class of simple 3-regular graphs on $n$ vertices, with $n\geq 4$ even.
The procedure shown in Fig.~\ref{fig1} defines Chain~I, a Markov chain
on~$\G_n$ parameterised by a probability $p\in(0,1)$.

Let $R$ be the desired number of steps. In the pseudocode below, the current state is denoted by $G=(V,E)\in\G_n$.
\begin{figure}[H]
\begin{tabbing}
aaa \= aa \= aa \= aa \= aa \kill
\tabrule\\
\bf triangle process$(p,1-p)$ \quad (Chain~I)\\[0.5ex]
\textbf{repeat} $R$ times\\
\>choose $v \in V$ and an unordered pair of two neighbours $x,w$ of $v$ u.a.r.\\
\> \textbf{if} $xw\in E$\\
\>\>choose an oriented edge $yz\in E$ u.a.r.\\
\>\>\textbf{if} $\{y,z\}\cap\{v,x,w\}=\es$ and $xy,wz\notin E$\\
\>\>\>with probability $1-p$\\
\>\>\>\>$\bk(vxw,yz)$\\
\>\> \textbf{else} do nothing\\
\> \textbf{else} choose $y\in N_G(x)\setminus \{v\}$ and $z \in N_G(w)\setminus \{v\}$ u.a.r.\\
\> \> \textbf{if} $y\neq z$ and $yz\notin E$\\
\>\>\> with probability $p$\\
\>\>\>\> $\mk(yxvwz)$\\
\>\> \textbf{else} do nothing\\
{\bf end}\\
\tabrule
\end{tabbing}
\caption{\;\;  Chain I: A proposed make move, if valid, is only performed with probability $p$.}
\label{fig1}
\end{figure}

Observe that Chain~I is aperiodic, since the self-loop (``do nothing'') occurs with positive probability at any step.
We have shown in Section~\ref{ergo:sec} that Chain~I is irreducible.
This implies that Chain~I is ergodic, and so has a unique stationary distribution on $\G_n$.

\subsection{Chain~II}\label{number:sec}\label{proc2:sec}

Chain~II makes transitions based on the number of triangles at the chosen vertex.
Chain~II uses the same type of make and break moves as Chain~I above. The make moves are however sampled  u.a.r. from the set of all feasible make moves at the selected vertex.

For each vertex $v$ in the current graph $G$, let ${\cal Q}_v$ be the set of unordered pairs of distinct paths of length~2 from $v$ which
define valid make moves. That is,
\begin{align*}
{\cal Q}_v &= \bigl\{ \{vxy, vwz\}: \, x,w\in N_G(v),\, x\neq w,\, wx \not \in E(G),\, y\in N_G(x)\setminus \{v\},\,\\
  & \hspace*{8cm}
       z\in N_G(w)\setminus \{v\},\,   y \ne z,\,  yz \not \in E(G)\bigr\}.
\end{align*}
It follows from the proof of Theorem~\ref{TH1} that if $\D_v \le 2$ then ${\cal Q}_v \ne \emptyset$.
We refer to elements of ${\cal Q}_v$ as path-pairs.

\newpage  

The pseudocode for Chain~II is shown in Fig.~\ref{fig2} below.  Again, the current state at
each step is denoted by $G=(V,E)$, with arbitrary starting state.

\vspace{-0.1in}
\begin{figure}[H]
\begin{tabbing}
aaa \= aa \= aa \= aa \= aa \kill
\tabrule\\
{\bf Triangle process (Chain II)}\\
{\bf repeat} $R$ times\\
\>choose $v \in V$ u.a.r.\\
\>let $\D_v$ be the number of triangles at $v$\\
\>with probability $\D_v/3$ \\
\>\>choose a u.a.r.\ (oriented) triangle $vxwv$ at $v$\\
\>\>choose an oriented edge $yz\in E$ u.a.r.\\
\>\>\textbf{if} $\{y,z\}\cap\{v,w,x\}=\es$ and $xy,wz\notin E$\\
\>\>\>\>$\bk(vxw,yz)$\\
\>\> \textbf{else} do nothing\\
\>\textbf{otherwise} (with probability $1-\D_v/3$)\\
\>\> choose $\{vxy,vwz\}\in {\cal Q}_v$ u.a.r.\\
\>\> $\mk(yxvwz)$\\
{\bf end}\\
\vspace{-0.05in}
\tabrule
\end{tabbing}
\caption{\;\; Chain II. Make and break probabilities depend on the number of local triangles.}\label{fig2}
\end{figure}

As some break moves will be rejected in any graph, there is a self-loop
on every state and hence Chain~II is aperiodic. It follows that Chain~II
is ergodic, by Theorem~\ref{TH1}, and hence has a unique stationary distribution.
The vertex $v$ which is chosen first in a make move $\mk(yxvwz)$ is called the {\em central vertex}
of the make move.


\section{The long-run number of triangles} \label{trisec}

In this section, we prove Theorem~\ref{TH2} for Chain~II, and  Corollary \ref{coTH2} for Chain~I.

Let $\D(G)$ denote the number of triangles in $G\in\G_n$.
Let $S(G)$ be a set of interest: it may be the set of vertices of $G$ in triangles,
or the set of triangles in $G$, depending on the context. Let $\mu_S$ be the expected change in the size of
$S$ after one step of Chain II.
Thus
\[
\mu_S(G)= \E\Big(\, |S(G')|\, \Big| \, S(G)\, \Big)- |S(G)|.
\]
Intuitively, the size of $S$ should tend to increase if $\mu_S>0$, and decrease if $\mu_S<0$. This is formalized in Lemma~\ref{Drift}.

\subsection{Lower bound on number of triangles, Chain~II}

For a lower bound, let $S(G)$ be the set of all vertices of $G$ which belong to at least one triangle.
Then $\Delta(G) \geq |S(G)|/3$.
For ease of notation we write $S$, $\Delta$ rather than $S(G)$, $\Delta(G$).
Also define $S' = N(S)\setminus S$, the set
of vertices which do not belong to a triangle but are adjacent to at least one triangle.  Finally, let
$F' = V(G)\setminus (S\cup S')$ be the set of all vertices which do not belong to a triangle and are not adjacent to any triangle.
Observe that $|S'|\leq |S|$ as each vertex in $S$ can have at most one neighbour outside $S$.

To assess the effect of a break move on a vertex of $S$, we argue as follows.
Let $T$ be the set of tetrahedra, and let $W$ be the set of vertices on tetrahedra. Thus $|W|=4|T|$.
Then each vertex in $S \setminus W$ must belong to a diamond or
to an isolated triangle.
Each tetrahedron spans 6 edges, each diamond 5 edges and each isolated triangle 3 edges.
Let $E(S)$ be the of edges in triangles, and define
\[ s=\frac{|S|}{n}, \qquad s'=\frac{|S'|}{n}, \qquad  t=\frac{|W|}{n}, \qquad e(S)=\frac{|E(S)|}{n}. \]
Then
\[ e(S)\leq  \frac{6t}{4} + \frac{5(s-t)}{4} = \frac{5s+t}{4}.\]
Now consider the options for the additional edge $yz$ which is used in a break move.
The probability that the oriented edge $yz$ belongs to a triangle is
\[ \frac{2 |E(S)| }{3n} = \frac{2\, e(S)}{3} \le \frac{5s+t}{6}.
\]
The probability that the oriented edge $yz$ belongs to a tetrahedron is
\[ \frac{|W|}{4}\times 12 \times \frac{1}{3n} = \frac{|W|}{n}=t,\]
and therefore, the probability that $yz$ is an edge of a triangle but not an edge of a tetrahedron is at most
$5(s-t)/6$.

The table below gives the worst case reduction in the number of vertices in triangles when the
break move $\bk(vxw,yz)$ is performed, characterised by whether the triangle $vxwv$ belongs to a tetrahedron,
and whether the edge $yz$ belongs to a tetrahedron, or belongs to a triangle but not a tetrahedron,
or does not belong to a triangle.

\begin{center}
\renewcommand{\arraystretch}{1.2}

\begin{tabular}{|l||c|c|c|}
\hline
&\multicolumn{3}{c|}{{\bf edge} $yz$}\\
\hline
{\bf triangle} $vxwv$  &$E(W)$&$E(S\setminus W)$&$E(G)\setminus E(S)$\\
\hline
\hline
tetrahedron &0&$\le 4$&0\\
\hline
diamond or isolated triangle&$\le 4$&$\le 8$&$\le 4$\\
\hline
\end{tabular}

\end{center}

\medskip
For example, the worst case is when both $xy$ and $yz$ are diagonal edges of distinct diamonds. The diagonal
edge of a diamond is the unique central edge contained in both triangles of the diamond.
After performing $\bk(vxw,yz)$, all 8 vertices involved in these two diamonds are now no longer contained in any
triangle.

\bigskip

Next we consider the effect of make moves on the number of vertices contained in triangles.
Write $S'=A \cup B$, where $A$ are vertices of $S'$ with exactly one neighbour contained in a triangle, and
$B$ are those vertices in $S'$ with at least two neighbours contained in a triangle.  Thus $s'=a+b$, where
$a=|A|/n$ and $b=|B|/n$.
Counting edges from $S \setminus W$ to $S'$, we have
\begin{equation}
\label{needed}
 a+2b \le s-t.
\end{equation}
\begin{itemize}
\item If a vertex of $A$ is chosen as $v$ and a make move is performed, with $v$ in the centre of the 4-path,
 then at most one triangle is broken by the make move.  The number of vertices in triangles cannot decrease in this case.
\item If a vertex of $B$ is chosen as $v$ then the overall number of triangles reduces by at most one, and hence the
number of vertices in triangles reduces by at most 3.
\item By Lemma~\ref{ergo:lem10}, if a vertex of $F'$ (a vertex not in or adjacent to any triangle) is chosen, at least one extra triangle is inserted,
and hence the number of extra vertices in triangles is at least 3.
The probability of this event is $1-(s+s')$.
\item If a vertex in a triangle is chosen, in all cases we assume  triangle breaking occurs. The number of triangle vertices lost in worst case is given in the above table. Inserting a  triangle at a vertex
    already in a triangle, can break at most two existing triangles, in which case at
     most 3 triangle vertices are lost. This  is less than the worst case values used for the triangle breaking case.
    \end{itemize}
By combining all these worst case bounds, we have
\begin{flalign}
\mu_S(G) & \ge 3\big(1-(s+s')\big) -3b  -\frac{ 4t\times 5(s-t)}{6}  \label{ssb}\\
& \qquad \qquad {} -4(s-t)t  - \frac{8(s-t)\times 5(s-t)}{6}  -4(s-t)\left(1-\frac{5s+t}{6}\right)\nonumber\\
&= \left(3-{6s}\right)+\brac{-4s -\frac{10s^2}{3}+ t\brac{7+\frac{10s}{3}}} \nonumber\\
& \ge 3- {10s}-\frac{10}{3}s^2=f(s). \label{MusLB}
\end{flalign}
The inequality (\ref{needed}) implies $s'+b =a+2b\le s-t$. This is used as follows to simplify \eqref{ssb},
\[
3(1 - (s+s')) - 3b \geq 3(1-2s+t) \geq 3 - 6s.
\]
The quadratic $f(s)=3- 10s-10s^2/3$ has one negative root and a positive root at $s^+>0.2748$.
Since $\D(G) \ge |S(G)|/3$,  it follows that
the expected change in the number of triangles after one step of Chain~II is positive whenever $\D(G) < s^+n/3$.
We prove in Lemma~\ref{Drift} below that w.h.p.\ $\D \ge 0.09n$ in the long run.

\subsection{Upper bound on number of triangles, Chain~II }

Suppose that $G$ is the current graph, and let
$\D,I,D,T$ be the number of triangles, isolated triangles, diamonds and tetrahedra
in $G$, respectively.
Let $F$ be the number of free vertices in $G$; that is, vertices not contained in any triangle.
Thus
\begin{equation}
\label{identities} \D=I+2D+4T, \qquad n=F+3I+4D+4T.
\end{equation}

Say that a vertex in a diamond is \emph{external} if it has degree~2 in the diamond, and
otherwise is \emph{internal} (degree~3 in the diamond).   The \emph{diagonal edge} in a diamond
is the edge which joins the two internal vertices of the diamond.
The next lemma gives an upper bound on the increase in the number of triangles as a result of a
make move, for different types of central vertex.


The following lemma, used in the upper bound proof, is proved in the next section.
\begin{lemma}
\label{lem:alpha}
From current graph $G\in\mathcal{G}_n$, let $v$ be a fixed vertex
and let $G'\in\mathcal{G}_n$ be the random graph obtained from $G$ by choosing $\{vxy,\, vwz\} \in \mathcal{Q}_v$
uniformly at random, and then performing the make move $\mk(yxvwz)$ to produce $G'$.
If $v$ is a free vertex (vertex in an isolated triangle, external diamond vertex, internal diamond vertex,
respectively) then the expected value of $\Delta(G') - \Delta(G)$ is bounded above by
\[
 \a_F = 8/3,\qquad \a_I = 3,\qquad \a_{D'} = 1,\qquad \a_{D''} = 4,
\]
respectively.
The expectation is with respect to the uniform distribution on $\mathcal{Q}_v$.
\end{lemma}

\bigskip

Let $G'$ be obtained from $G$ by one step of Chain~II which performs a make move.
Using (\ref{identities}) and Lemma~\ref{lem:alpha}, we have
\begin{flalign}
\E (\, \Delta(G') - \Delta(G) \, \mid  G, \mk)
& \le
\a_F \, \frac{F}{n} \, +\,\a_I \, \frac{3I}{n} \times \dfrac23 \, +\, \a_{D'} \, \frac{2D}{n} \times \dfrac23\, +\, \a_{D''} \, \frac{2D}{n} \times \dfrac13
\nonumber\\
&\leq \frac{8 F}{3n} + \frac{6I}{n} + \frac{4D}{n}
 \nonumber\\
&\leq \frac{8}{3} -\frac{2\D}{n}.\label{make-impact}
\end{flalign}
To explain the first inequality note that, for example, there are $3I$ vertices in isolated triangles.  For each, the probability that a make move will be performed is $2/3$, and an upper bound on the
increase in the number of triangles, averaged over all path-pairs in $\mathcal{Q}_v$,
 is given by $\alpha_I$.
The other terms are similar.

Next, we consider a break move $\bk(vxw,yz)$, where the triangle $vxwv$ is fixed and the
oriented edge $yz$ is chosen randomly from $E(G)$.
Some choices of $yz$ will be rejected; specifically, if $yz$ is one of the (oriented) edges in the
triangle, an oriented edge which is incident with the triangle, or when
at least one of the edges $xy$ or $wz$ are present in the current graph.
This gives at most $6 + 6 + 4= 16$ oriented edges $yz$ where the break move will be rejected.
Similarly, some choices of
the oriented $yz$ will lead to triangles being created by the break move.  This occurs when there is a path
of length 2 from $x$ to $y$ or from $w$ to $z$ (or both).  There are at most 8 such oriented edges $yz$,
and at most 2 triangles can be created.  (For two triangles to be created by the break move $\bk(vxw,yz)$
there must be a 6-cycle $xryzswx$, creating triangles $xryx$ and $wzsw$.)

%

Let $\g$ be the expected number of triangles which are destroyed due to the deletion of $yz$,
 as part of the break move $\bk(vxw,\cdot)$.
Then
\begin{align}
\g &=
\frac{2}{3n} \, \left(1\times 3I \, + \, 1 \times 4D \, +\,  2\times D  \, +\, 2 \times 6T\right)
- O(1/n) \nonumber \\
&= \frac{2\D}{n} - O(1/n).\label{fred}
\end{align}
The $O(1/n)$ term (with a positive implicit constant) arises from choices of $yz$ which are rejected
or which cause triangles to be created.  Since we are interested in $\D = \Theta(n)$, this term will
not be significant.
Hence, using (\ref{fred}),
 the expected effect of a break move on the number of triangles is given by
\begin{flalign}
 \E(\, \Delta(G)' &- \Delta(G) \, \mid G, \bk) \nonumber \\
 &= - \dfrac1n\brac{3I\times \nfrac13 (1+\g)+2D\times \nfrac13 (2+\g)+ 2D\times \nfrac23 (1+\g)+4T(2+\g)} \nonumber\\
&= - \dfrac1n\brac{I +\frac{8D}{3}+8T+\g(I+2D+4T)} \nonumber \\
&= - \dfrac1n\brac{\D+\frac{2D}{3}+4T+\g \D} \nonumber \\
&\leq -\frac{\D}{n} - \frac{2\D^2}{n^2} + O(\Delta/n^2).
\label{break-impact}
\end{flalign}
Now write $x=\D/n$ and recall that $0\leq \Delta\leq n$.
It follows from (\ref{make-impact}) and (\ref{break-impact}) that
\begin{align}
 \E(\, \Delta(G)' - \Delta(G) \, \mid G)
  &\leq \dfrac{8}{3} - 3x  - 2x^2 + O(\Delta/n^2)\nonumber \\
  &=\dfrac{8}{3} - 3x  - 2x^2 + O(1/n). \label{upper-quad}
\end{align}
This quadratic in $x$ becomes negative at $(-9 + \sqrt{273})/12 < 0.627$.
It will follow from Lemma~\ref{Drift} below that w.h.p.\ $\D \le 0.63n$
in the long run.
\bigskip

\subsection{Proof of Lemma~\ref{lem:alpha}}
For an upper bound, we will
assume that no triangles are accidently destroyed by the \tm.
For most cases, we will simply use an upper bound on the number of triangles which
can be created by any make move with $v$ as central vertex.
We will only perform careful averaging over $\mathcal{Q}_v$ in the case of free vertices,
which we treat last.

Suppose that the \tm $\mk(yxvwz)$ can be performed in the current graph $G$, and creates more than one triangle.
Then the path $yxvwz$ must be part of a hexagon (6-cycle), say $yxvwzuy$.
In a 3-regular graph, the maximum number of triangles which can be created during a \tm is 4, as shown in Fig.~\ref{fig:metro}.

\bigskip

\noindent {\em Case of $\a_{D''}$ (internal diamond vertex)}.\
Here we simply take $\a_{D''}=4$, which is always an upper bound on the
number of triangles created by a make move.  This upper bound can be
achieved, see Fig.~\ref{fig:metro}.

\bigskip

\noindent {\em Case of $\a_{D'}$ (external diamond vertex)}.\
Next suppose that $v$ is an external diamond vertex.  The diamond must share at least
one edge with the hexagon, or else the degree of $v$ will be at least~4.
Without loss of generality, assume that $vw$ is an edge
of the diamond.  Then $vx$ cannot also be
an edge of the diamond, for then $xw$ would be the diagonal edge of the diamond, but
we need $xw\not\in E$.
Similarly, if $vw$ is the only edge of the diamond which lies on the hexagon, then $w$ must have
degree at least~4 in $G$, as it has degree~3 in the diamond.  Hence the hexagon edge $wz$  must also be
an edge of the diamond.  Now there are two subcases.
If the edge $wz$ is not the diagonal edge of the diamond then the diagonal edge is $wa$ for some vertex
$a$ which does not lie on the hexagon (we cannot have $x=a$ since $wx\not\in E$, while $a\in \{y,u\}$ would
cause the degree of $y$ or $u$ to be too high).
Then the \tm $\mk(yxvwz)$ results in a graph $G'$ with $\Delta(G')-\Delta(G)\leq 1$.
See Fig.~\ref{alpha:external-diamond}(i).
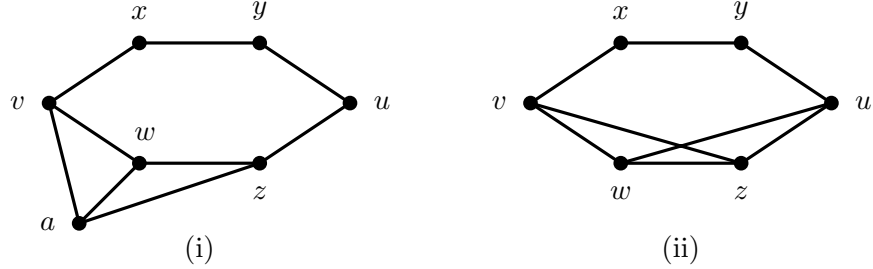
\begin{figure}[H]
\begin{center}
\begin{tikzpicture}[scale=0.8]
\draw [very thick,-] (5.5,1)--(4,0) -- (2,0) -- (0.5,1) -- (2,2) -- (4,2)--(5.5,1);
\draw [fill] (0.5,1) circle (0.1); \draw [fill] (2,0) circle (0.1);
\draw [fill] (2,2) circle (0.1);
\draw [fill] (4,2) circle (0.1); \draw [fill] (4,0) circle (0.1);
\draw [fill] (5.5,1) circle (0.1);
\draw [fill] (1.0,-1.0) circle (0.1);
\draw [very thick,-] (0.5,1) -- (1.0,-1.0)  (4,0) -- (1.0,-1.0) -- (2,0);
\node [left] at (0.3,1) {$v$}; \node [above] at (2.1, 0.2) {$w$};
\node [below] at (4,-0.2) {$z$}; \node [right] at (5.7,1.0) {$u$};
\node [above] at (4,2.2) {$y$}; \node [above] at (2,2.2) {$x$};
\node [left] at (0.8,-1.0) {$a$};
\node [below] at (3,-1.0) {(i)};
\begin{scope}[xshift=8cm]
\draw [very thick,-] (5.5,1)--(4,0) -- (2,0) -- (0.5,1) -- (2,2) -- (4,2)--(5.5,1);
\draw [fill] (0.5,1) circle (0.1); \draw [fill] (2,0) circle (0.1);
\draw [fill] (2,2) circle (0.1);
\draw [fill] (4,2) circle (0.1); \draw [fill] (4,0) circle (0.1);
\draw [fill] (5.5,1) circle (0.1);
\draw [very thick,-] (0.5,1) -- (4.0,0.0)  (2,0) -- (5.5,1.0);
\node [left] at (0.3,1) {$v$}; \node [below] at (2.0, -0.2) {$w$};
\node [below] at (4,-0.2) {$z$}; \node [right] at (5.7,1.0) {$u$};
\node [above] at (4,2.2) {$y$}; \node [above] at (2,2.2) {$x$};
\node [below] at (3,-1.0) {(ii)};
\end{scope}
\end{tikzpicture}
\end{center}
\caption{~~(i) $wz$ is not the diagonal edge, ~~(ii) $wz$ is the diagonal edge}\label{alpha:external-diamond}
\end{figure}
Otherwise, the edge $wz$ is the diagonal edge of the diamond.  If the fourth vertex of the diamond does not lie on the hexagon then $z$ must have degree at most~4, which is impossible.  Therefore the final vertex of the diamond
must be~$u$, and we have the situation shown in Fig.~\ref{alpha:external-diamond}(ii).  Here
the \tm $\mk(yxvwz)$ results in a graph $G'$ with $\Delta(G')-\Delta(G)\leq 0$.
This proves that we can take $\alpha_{D'}=1$.

\bigskip

\noindent {\em Case of $\a_I$ (isolated triangle)}.\
Suppose that $v$ belongs to an isolated triangle. Again, this isolated triangle
must share an edge with the hexagon, or $v$ will have degree at least~4.
Assume without loss of generality that
$vw$ is an edge in the isolated triangle.
If the third vertex of the triangle also lies on the hexagon then it must be $z$,
since $wx\not \in E$ and the triangle is assumed to be isolated:
see Fig.~\ref{alpha:isolated}(i).
Here $y$ and $z$ cannot have another common neighbour, and the \tm $\mk(yxvwz)$ results
in a graph $G'$ with $\Delta(G')-\Delta(G)\leq 1$.
\begin{figure}[H]
\begin{center}
\begin{tikzpicture}[scale=0.8]
\draw [very thick,-] (5.0,1)--(4,0) -- (2,0) -- (1,1)-- (2,2) -- (4,2) -- (5,1);
\draw [fill] (1.0,1) circle (0.1); \draw [fill] (2,0) circle (0.1);
\draw [fill] (2,2) circle (0.1);
\draw [fill] (4,2) circle (0.1); \draw [fill] (4,0) circle (0.1);
\draw [fill] (5.0,1) circle (0.1);
\node [left] at (0.8,1) {$v$}; \node [below] at (2.1, -0.2) {$w$};
\node [below] at (4,-0.2) {$z$}; \node [right] at (5.2,1.0) {$u$};
\node [above] at (4,2.2) {$y$}; \node [above] at (2,2.2) {$x$};
\node [below] at (3,-0.8) {(i)};
\draw [very thick,-] (1.0,1) -- (4.0,0.0);
\begin{scope}[xshift=10cm]
\draw [very thick,-] (5.0,1)--(4,0) -- (2,0) -- (1.0,1) -- (2,2) -- (4,2)--(5.0,1);
\draw [fill] (1.0,1) circle (0.1); \draw [fill] (2,0) circle (0.1);
\draw [fill] (2,2) circle (0.1);
\draw [fill] (4,2) circle (0.1); \draw [fill] (4,0) circle (0.1);
\draw [fill] (5.0,1) circle (0.1);
\node [left] at (0.8,1) {$v$}; \node [below] at (2.1, -0.2) {$w$};
\node [below] at (4,-0.2) {$z$}; \node [right] at (5.2,1.0) {$u$};
\node [above] at (4,2.2) {$y$}; \node [above] at (2,2.2) {$x$};
\node [left] at (-0.2,0.0) {$a$};
\node [below] at (3.0,-0.8) {(ii)};
\draw [-,dashed] (4,2)--(3,1)-- (5,1) (3,1) -- (4,0);
\draw [fill=white] (3,1) circle (0.1);
\draw [fill] (0.0,0.0) circle (0.1);
\draw [very thick,-] (1.0,1) -- (0.0,0) -- (2,0);
\end{scope}
\end{tikzpicture}
\end{center}
\caption{~~$v$ and $w$ have a common neighbour which lies (i) on, ~~(ii) off the hexagon.}\label{alpha:isolated}
\end{figure}
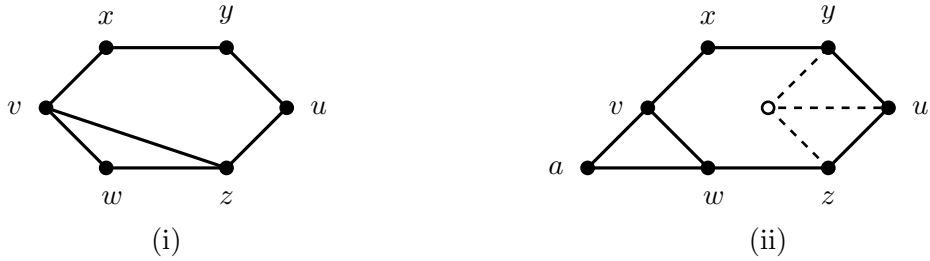
Otherwise, the third vertex of the triangle $a$ is distinct from the
vertices of the hexagon, as in Fig.~\ref{alpha:isolated}(ii).
Then the \tm $\mk(yxvwz)$ forms a graph $G'$
$\Delta(G') - \Delta(G) \leq 3$, so we can take $\a_I=3$.


\bigskip

\noindent {\em Case of $\a_F$ (free vertex)}.\
When $v$ is a free vertex we will give a more careful analysis,
averaging the number of triangles formed over elements of $\mathcal{Q}_v$.
First suppose that an increase of~3 triangles is possible at $v$. This implies the
presence of the subgraph shown in Fig.~\ref{hexhex}(i), which we call a singly-augmented
hexagon centred at $v$. (Note that an increase of 3~triangles is obtained when $u$ and
$v$ exchange labels, so the singly-augmented hexagon is centred at $u$ and $v$ is in the
4-cycle.)
Similarly, an increase of~4 at $v$ implies the presence of the
subgraph shown in Fig.~\ref{hexhex}(ii),  which we call a doubly-augmented hexagon on $v$.
Let $N_r(v)$ denote the vertices which are at distance~$r$ from $v$, for any positive
integer $r\geq 2$.
As there are no triangles at $v$, there are no edges in $N(v)=\{x,y,z\}$, and thus 6 edges
from $N(v)$ to $N_2(v)$.
In a 3-regular graph, the vertex $v$ can be in no more than 6 hexagons (corresponding to the 12 path-pairs
of length two from $v$).

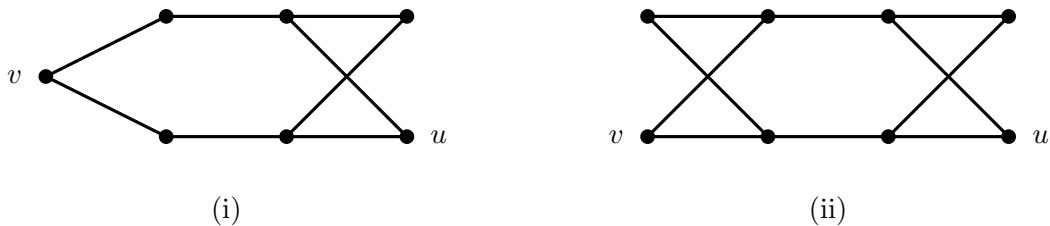
\begin{figure}[H]
\begin{center}
\begin{tikzpicture}[scale=0.8]
\draw [very thick,-] (6.0,0)--(4,0) -- (2,0) -- (0.0,1) -- (2,2) -- (4,2)--(6.0,0);
\draw [very thick,-] (4,2)-- (6.0,2)--(4,0) ;
\draw [fill] (0.0,1) circle (0.1); \draw [fill] (2,0) circle (0.1);
\draw [fill] (2,2) circle (0.1);
\draw [fill] (6,2) circle (0.1);
\draw [fill] (4,2) circle (0.1); \draw [fill] (4,0) circle (0.1);
\draw [fill] (6.0,0) circle (0.1);
\node [left] at (-0.2,1) {$v$};
\node [right] at (6.2,0.0) {$u$};
\node [below] at (3,-0.8) {(i)};
\begin{scope}[shift={(10,0)}]
\draw [very thick,-] (6.0,0)--(4,0) -- (2,0) -- (0.0,0) -- (2,2) -- (4,2)--(6.0,0);
\draw [very thick,-] (4,2)-- (6.0,2)--(4,0)  (2,0) -- (0.0,2) -- (2,2);
\draw [fill] (0.0,0) circle (0.1); \draw [fill] (2,0) circle (0.1);
\draw [fill] (2,2) circle (0.1);
\draw [fill] (0,2) circle (0.1); \draw [fill] (6,2) circle (0.1);
\draw [fill] (4,2) circle (0.1); \draw [fill] (4,0) circle (0.1);
\draw [fill] (6.0,0) circle (0.1);
\node [left] at (-0.2,0) {$v$}; 
 \node [right] at (6.2,0.0) {$u$};
\node [below] at (3,-0.8) {(ii)};
\end{scope}
\end{tikzpicture}
\end{center}
\caption{~~ (i) Singly-augmented hexagon. (ii) Doubly-augmented hexagon.}\label{hexhex}
\end{figure}

Given three edge-disjoint paths of length two rooted at $v$,
we can form at most three singly-augmented hexagons centred at $v$ by taking
these three paths and joining each of the three end-vertices to two additional vertices.
See Fig.~\ref{sah}(i).   This construction can be ``doubled'' by taking the other
three paths of length two rooted at $v$, assuming they are also edge-disjoint, leading
to the graph shown in Fig.~\ref{6sah} below.  This shows that $v$ can be at the centre
of up to six singly-augmented hexagons.

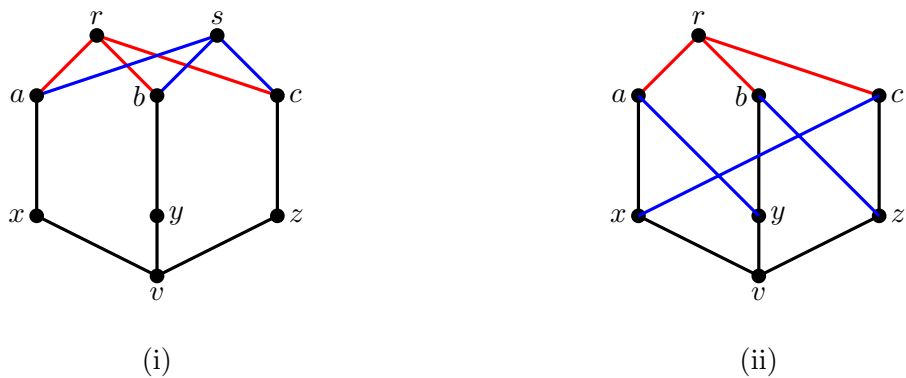
\begin{figure}[H]
\begin{center}
\begin{tikzpicture}[scale=0.8]

\draw [very thick,-] (4,3)--(4,1)--(2,0)--(2,1)--(2,3)
(2,0)--(0,1)--(0,3);
\draw [very thick,-,red] (0,3)--(1,4)--(2,3) (1,4)--(4,3);
\draw [very thick,-,blue] (0,3)--(3,4)--(2,3) (3,4)--(4,3);
\node [below] at (2,0) {$v$};
\node [above] at (1,4) {$r$};
\node [above] at (3,4) {$s$};

\draw [fill] (4,3) circle (0.1);\draw [fill] (4,1) circle (0.1);\draw [fill] (2,0) circle (0.1);\draw [fill] (2,1) circle (0.1);\draw [fill] (2,3) circle (0.1);\draw [fill] (0,1) circle (0.1);\draw [fill] (0,3) circle (0.1);
\draw [fill] (1,4) circle (0.1);\draw [fill] (3,4) circle (0.1);
\node [below] at (2,-1) {(i)};

\node [right] at ( 4,1 ) {$z$};
\node [right] at ( 2,1 ) {$y$};
\node [left] at ( 0,1 ) {$x$};
\node [right] at ( 4,3  ) {$c$};
\node [left] at ( 0,3 ) {$a$};
\node [left] at ( 2,3 ) {$b$};

\begin{scope}[xshift=10cm]
\draw [very thick,-] (4,3)--(4,1)--(2,0)--(2,1)--(2,3)
(2,0)--(0,1)--(0,3);
\draw [very thick,-,red] (0,3)--(1,4)--(2,3) (1,4)--(4,3);
\node [below] at (2,0) {$v$};
\draw [fill] (4,3) circle (0.1);\draw [fill] (4,1) circle (0.1);\draw [fill] (2,0) circle (0.1);\draw [fill] (2,1) circle (0.1);\draw [fill] (2,3) circle (0.1);\draw [fill] (0,1) circle (0.1);\draw [fill] (0,3) circle (0.1);
\draw [fill] (1,4) circle (0.1);
\draw [very thick,-,blue] (0,1)--(4,3)(0,3)--(2,1)(2,3)--(4,1);
\node [above] at (1,4) {$r$};
\node [below] at (2,-1) {(ii)};
\node [right] at ( 4,1 ) {$z$};
\node [right] at ( 2,1 ) {$y$};
\node [left] at ( 0,1 ) {$x$};
\node [right] at ( 4,3  ) {$c$};
\node [left] at ( 0,3 ) {$a$};
\node [left] at ( 2,3 ) {$b$};
\end{scope}
\end{tikzpicture}
\end{center}
\caption{~~ Maximal structures based on three paths of length two from $v$.}\label{sah}
\end{figure}
The maximum number of doubly-augmented hexagons on $v$ is three.
This maximum is obtained using the subgraph shown in Fig.~\ref{sah}(ii),
which again is based on three edge-disjoint paths of length 2 rooted at $v$
(the black edges).
Here there are only three vertices in $N_2(v)$, and the edges between $N(v)$ and
$N_2(v)$ form a 6-cycle (alternating blue and black edges).
Then there are still three 4-cycles at $r$, and the same blue edges now also form three
4-cycles at $v$. The graph is redrawn in Fig~\ref{hchex}.
Note that this forms a connected 3-regular component on 8 vertices.

Let $i$ (respectively, $j$, $\ell$, $m$) denote the number of path-pairs from $v$ which
create 4 triangles (respectively, 3 triangles, 2 triangles, 1 triangle) under a make move.
Given the counts $(i,j,\ell,m)$, which describe the structure of the subgraph at $v$,
the expected increase in the number of triangles under a make move using a random
path-pair from $\mathcal{Q}_v$ is equal to
\[
\psi'(i,j,\ell,m)=\frac{4i+3j+2\ell+m}{i+j+\ell+m}.
\]
For simplicity, in many sub-optimal cases we  let $k=\ell+m$ and replace
$\psi'(i,j,\ell,m)$ by the upper bound
\[
\psi(i,j,k)= \frac{4i+3j+2k}{i+j+k}.
\]
We now examine the functions $\psi'(i,j,\ell,m)$ and $\psi(i,j,k)$. We claim that
\[
\max_{i,j,\ell,m}\, \psi'(i,j,\ell,m)=\psi'(3,0,6,0)=\psi(3,0,6)=\dfrac{8}{3},
\]
with the maximum achieved by the subgraph shown in Fig.~\ref{hchex}(i).
To prove this is true, we  show that all other cases satisfy $\psi(i,j,k)\le 8/3$.

Let $s=i+j+\ell+m = |\mathcal{Q}_v|$ be the number of feasible path-pairs, which
is determined by the exact structure of the subgraph around $v$.  Note that
$k=s-i-j$, and hence the
condition $\psi(i,j,k)\leq 8/3$ is equivalent to the inequality
\begin{equation}\label{constrainij}
2i+j \leq  \frac{2s}{3}.
\end{equation}

We will give a case analysis based on $N$, which we define to be the number of 4-cycles
at $v$.  Note $N$ is not directly related to $i$, as some arrangements of 4-cycles do not
permit doubly-augmented hexagons.  The aim is to maximize $2i+j$ given $N$,
and to show that (\ref{constrainij}) holds for this maximum.

Firstly observe that while $N\in \{4,5,6\}$ is possible (for example,
$N=6$ if $v$ is in a $K_{6,6}$ component),
the largest value of $N$ which can occur when $v$ is contained in a singly-  or
doubly-augmented hexagon is $N=3$.  Therefore, if $N > 3$ then $\psi(i,j,k)\leq 2$
and (\ref{constrainij}) holds.
So we focus on $N\leq 3$.
We say that vertex $w \in N_2(v)$ has \emph{free degree} $x$ if there are $3-x$ edges from
$N(v)$ to $w$.

\paragraph{\em Case $N=3$. Three 4-cycles at $v$.}

Suppose that there are three 4-cycles at $v$, and label them $C_1$, $C_2$, $C_3$.
We show that it is always possible to form three doubly-augmented hexagons on $v$,
giving $i=3$ (which is the maximum possible for $i$, as noted earlier).
\begin{itemize}
\item
First suppose that each of the pairwise intersections
$C_1\cap C_2$, $C_2\cap C_3$, $C_1\cap C_3$ is a single edge incident
with $v$.  We can extend this to form the subgraph
shown in Fig. \ref{hchex}(i) below, redrawn from Fig. \ref{sah}(ii).  This is the
3-dimensional hypercube, which has $(i,j,\ell,m) = (3,0,6,0)$ and hence satisfies
$\psi'(3,0,6,0) = \psi(3,0,6)=8/3$.  There are 9 feasible path-pairs in total,
and the three path-pairs which produce 4~triangles
under a make move are $\{ vxa,\, vzb\}$, $\{ vxc,\, vyb\}$ and $\{ vya,\, vzc\}$.
\item
Otherwise, each of the pairwise intersections of the 4-cycles consists of a
2-path rooted at $v$,  and $C_3$ is the symmetric difference of $C_1$ and $C_2$.
This structure can be extended to the 10-vertex graph shown in Fig.~\ref{hchex}(ii),
which again has 9 feasible path-pairs; here
$(i,j,\ell,m) = (3,0,0,6)$ and hence $\psi'(3,0,0,6) = 2$.
\end{itemize}

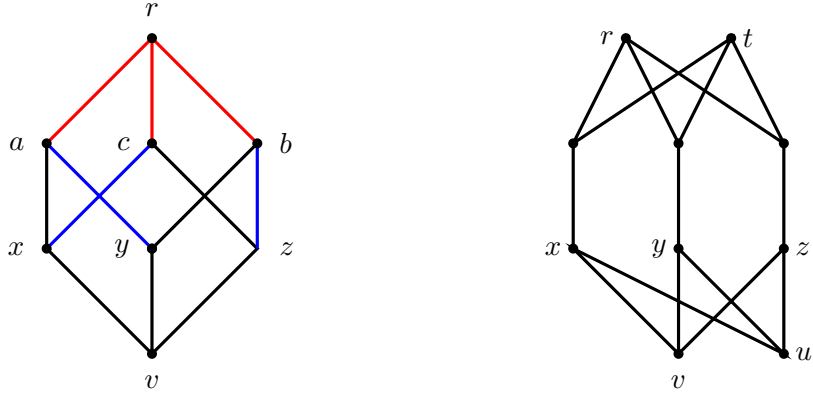
\begin{figure}[H]
\begin{center}
\begin{tikzpicture}[scale=0.7]
\draw[very thick, -] (2,-2)--(0,0)--(0,2) (2,2)--(4,0)--(2,-2);
\draw[very thick, -] (2,-2)--(2,0)--(4,2);
\draw[very thick,red] (0,2)--(2,4)--(2,2) (2,4)--(4,2);
\draw[very thick,blue] (0,0)--(2,2)(0,2)--(2,0) (4,2)--(4,0);
\draw [fill] (2,-2) circle (0.07);
\draw [fill] (0,0) circle (0.07);
\draw [fill] (0,2) circle (0.07);
\draw [fill] (2,4) circle (0.07);
\draw [fill] (2,0) circle (0.07);
\draw [fill] (2,2) circle (0.07);
\draw [fill] (2,0) circle (0.07);
\draw [fill] (4,2) circle (0.07);
\node [below] at (2, -2.2) {$v$};
\node [left] at (1.8,2) {$c$};\node [above] at (2, 4.2) {$r$};
\node [left] at (1.8, 0) {$y$};
\node [left] at (-0.2,0) {$x$}; \node [left] at (-0.2,2) {$a$};
\node [right] at (4.2,0.0) {$z$};\node [right] at (4.2,2.0) {$b$};
\begin{scope}[xshift=10cm]
\draw[very thick] (2,-2)--(0,0)--(4,-2)--(2,0)--(2,-2)--(4,0)--(4,-2);
\draw[very thick] (0,0)--(0,2)--(1,4)--(2,2)--(3,4)--(0,2);
\draw[very thick] (4,2)--(1,4) (4,2)--(3,4) (2,0)--(2,2) (4,0)--(4,2);
\node [below] at (2, -2.2) {$v$};
\node [right] at (4,-2) {$u$};
\node [left] at (1,4) {$r$};
\node [right] at (3,4) {$t$};
\node [right] at (4,0) {$z$};
\node [left] at (2,0) {$y$};
\node [left] at (0,0) {$x$};
\draw [fill] (1,4) circle (0.07);
\draw [fill] (2,-2) circle (0.07);
\draw [fill] (0,0) circle (0.07);
\draw [fill] (2,0) circle (0.07);
\draw [fill] (2,2) circle (0.07);
\draw [fill] (4,0) circle (0.07);
\draw [fill] (4,2) circle (0.07);
\draw [fill] (0,2) circle (0.07);
\draw [fill] (4,-2) circle (0.07);
\draw [fill] (3,4) circle (0.07);
\end{scope}
\end{tikzpicture}
\end{center}
\caption{~~$N=3$:\ (i)
 A 3-dimensional hypercube. (ii) Component with 10 vertices, where the three 4-cycles at $v$ each share two edges.
}\label{hchex}
\end{figure}

\paragraph{\em Case $N=2$. Two 4-cycles at $v$.}\  Let the 4-cycles at $v$ be $C_1$, $C_2$.
We saw above that if $C_1\cap C_2$ is a 2-path rooted at $v$
then the symmetric difference of $C_1$ and $C_2$ is another
4-cycle which contains $v$.  Hence $C_1\cap C_2$ is a single edge incident with $v$ (as
$N=2$, by assumption).

We claim that $i=0$, that is, it is not possible to
extend $C_1\cup C_2$ to a subgraph which contains a doubly-augmented hexagon on $v$.
See for example Fig.~\ref{dah}(i), where the edge in $C_1\cap C_2$ is $vy$ and
$c\neq d$ (or else $N=3$).
To construct a doubly-augmented hexagon using the 4-cycle $vxayv$, say, we must use
the edges $xd, yb$. But $b$ already has 2 neighbours in $N(v)$ (that is, $b$ has
free degree~1), so it is impossible to
add a 4-cycle $drbtd$ to complete the doubly-augmented hexagon.
By symmetry, this shows that $i=0$ as claimed.
\begin{figure}[ht!]
\begin{center}
\begin{tikzpicture}[scale=0.8]
\draw [very thick,-] (4,3)--(4,1)--(2,-1)--(2,1)--(2,3) (2,-1)--(0,1)--(0,3);
\node [below] at (2,-1) {$v$};
\draw [fill] (4,3) circle (0.1);\draw [fill] (4,1) circle (0.1);
\draw [fill] (2,-1) circle (0.1);\draw [fill] (2,1) circle (0.1);
\draw [fill] (2,3) circle (0.1);\draw [fill] (0,1) circle (0.1);
\draw [fill] (0,3) circle (0.1);
\draw [very thick,-] (0,3)--(2,1)(2,3)--(4,1);
\draw[very thick] (-2,3)--(0,1);
\draw [fill] (-2,3) circle (0.1);
\node [left] at (-2,3   ) {$d $};
\node [right] at ( 2,3  ) {$b $};
\node [right] at (0,3  ) {$a $};
\node [right] at ( 4,3  ) {$c$};
\node [left] at ( 0,1  ) {$x $};
\node [right] at ( 2,1 ) {$ y$};\node [right] at ( 4,1  ) {$z$};
\begin{scope}[shift={(10,0)}]
\draw[very thick, blue] (0,3) -- (2,5) (0,5) -- (2,3); 
\draw[very thick, red] (4,3)--(0,5)--(-2,3) -- (2,5) -- (4,3); 
\draw [very thick,-] (4,3)--(4,1)--(2,-1)--(2,1)--(2,3) (2,-1)--(0,1)--(0,3);
\node [below] at (2,-1) {$v$};
\draw [fill] (4,3) circle (0.1);\draw [fill] (4,1) circle (0.1);\draw [fill] (2,-1) circle (0.1);\draw [fill] (2,1) circle (0.1);\draw [fill] (2,3) circle (0.1);\draw [fill] (0,1) circle (0.1);\draw [fill] (0,3) circle (0.1);
\draw [very thick,-] (0,3)--(2,1)(2,3)--(4,1);
\draw[very thick] (-2,3)--(0,1);
\draw [fill] (-2,3) circle (0.1);\draw [fill] (2,5) circle (0.1);
\draw [fill] (0,5) circle (0.1);
\node [above] at (0,5) {$r$};
\node [above] at (2,5) {$t$};
\node [left] at (-2,3   ) {$d $};
\node [right] at ( 2,3  ) {$b $};
\node [right] at (0,3  ) {$a $};
\node [right] at ( 4,3  ) {$c$};
\node [left] at ( 0,1  ) {$x $};
\node [right] at ( 2,1 ) {$ y$};\node [right] at ( 4,1  ) {$z$};
\end{scope}
\end{tikzpicture}
\end{center}
\caption{~~$N=2$.\quad (i) All 2-paths from $v$.\quad (ii) The maximising graph.}\label{dah}
\end{figure}
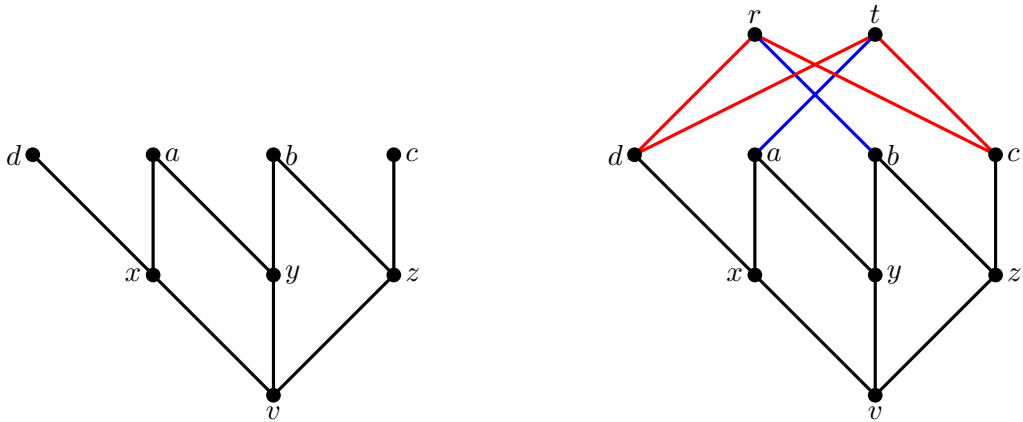

Indeed, only vertices $c,d\in N_2(v)$ have free degree two, so there can be at most one 4-cycle in
$N_2(v)\cap N_3(v)$.  Suppose that $drctd$ is such a 4-cycle involving two new vertices $r$, $t$
(shown as the red edges in Fig.~\ref{dah}(ii).
Both $a$ and $b$ still require one more incident edge.
\begin{itemize}
\item If we join $a$ to $t$ and $b$ to $r$, for example, then we obtain the
graph shown in Figure~\ref{dah}(ii).  Here $j=3$ and $s=10$, and hence (\ref{constrainij})
holds as $2j = 6 < 20/3  = 2s/3$.
\item If we insert $ab$ instead then we obtain a graph with $j=1$ and $s=7$,
with feasible path-pairs 
\[ \{ vxd,\, vya\}, \{ vxd,\, vyb\}, \{vxd,\, vzb\}, \{vxd, vzc\}, \{vzc,\, vxa\},
   \{vxc,\, vya\},\, \{vxc,\, vyb\}.\]
Again (\ref{constrainij}) holds as $2j = 2 < 14/3 = 2s/3$.
\item No other choice gives a positive value of $2i+j$.
\end{itemize}



\paragraph{\em Case $N=1$. One 4-cycle at $v$.}\

When $N=1$, the maximum possible value of $i$ that we can achieve is $i=1$.
Fig.~\ref{1aughex}(i) shows a typical construction which will give $N=i=1$.
There are four vertices in $N_2(v)$ with free degree~2, namely $a$, $c$, $d$ and $e$.
To maximise $2i+j$, we may assume that the 4-cycle $atcra$ is present,
and then the path-pair $\{ vxa, vyc\}$ gives an increase of 4 triangles after a make move.
It remains to consider options for the remaining edges which are incident with
vertices $d$ and $e$, with free degree two, and $b$ with free degree one.

It is convenient to use a condensed notation for path-pairs, given by $(\a; (\b,...,\g))$. This enumerates the feasible (i.e., switchable) path-pairs, one of which is $\a$. As an example, in Fig.~\ref{dah}, $(d;(ya,yb,zb,c))$ are the feasible pairs for the unique path ending at $d$ (i.e. $vxd$) and the other paths. We list as much of the path as required for uniqueness. Thus $ya$ distinguishes $vya$ from $vxa$ etc.

\begin{itemize}
\item We can join $d$ to $r$ and $t$ to give $i=1$, $j=2$.  Now to maximise $\psi$
we want to minimise $s$, the number of valid path-pairs.  To achieve this, we can add
an edge from $b$ to $e$.  This leads to the subgraph shown in Fig.~\ref{1aughex}(ii).
Here $s=9$, with path-pairs
$(a;(yb,c,d,e)),\; (d;(xb,yb)),\;(c;(xb,d,e))$.
No further reduction in $s$ is possible for $i=1,j=2$.  Since $2i+j = 4 < 6 = 2s/3$,
condition (\ref{constrainij}) holds.
\item In all other cases we will have $i=1$, $j=0$.  Since we have at least 6 valid
path-pairs, namely $(a;(xb,c,d,e))$, $(c,(d,e))$, we see that (as $2 < 4 \leq 2s/3$, (\ref{constrainij})
holds.
\end{itemize}

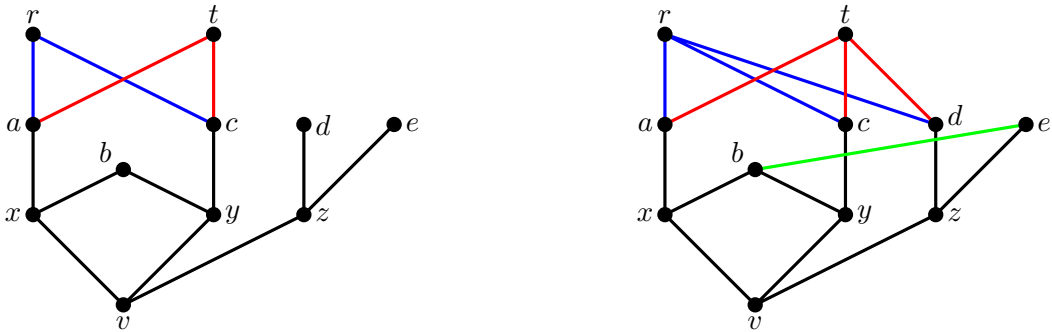
\begin{figure}[H]
\begin{center}
\begin{tikzpicture}[scale=1.2]
\draw[very thick, -] (2,0)--(1,1)--(2,1.5)--(3,1)--(2,0) (2,0)--(4,1);
\draw[very thick, -] (1,1)--(1,2) (3,2)--(3,1) ;
\draw[very thick, -, blue] (1,2)--(1,3)--(3,2); 
\draw[very thick, -, red] (1,2)--(3,3)--(3,2); 
\draw[very thick, -] (4,2)--(4,1)--(5,2);
\node [below] at (2,0) {$v$};
\node [left] at (1,1) {$x$};
\node [right] at (3,1) {$y$};
\node [right] at (4,1) {$z$};
\node [left] at (1,2) {$a$};
\node [left] at (2,1.7) {$b$};
\node [right] at (3,2) {$c$};
\node [right] at (4,2) {$d$};
\node [right] at (5,2) {$e$};
\node [above] at (3,3) {$t$};
\node [above] at (1,3) {$r$};
\draw [fill] (1,1) circle (0.07); \draw [fill] (1,2) circle (0.07);
\draw [fill] (2,0) circle (0.07); \draw [fill] (2,1.5) circle (0.07);
\draw [fill] (1,3) circle (0.07); \draw [fill] (3,1) circle (0.07);
\draw [fill] (3,2) circle (0.07); \draw [fill] (3,3) circle (0.07);
\draw [fill] (4,1) circle (0.07); \draw [fill] (4,2) circle (0.07);
\draw [fill] (5,2) circle (0.07);
\begin{scope}[shift={(7,0)}]
\draw[very thick, -] (2,0)--(1,1)--(2,1.5)--(3,1)--(2,0) (2,0)--(4,1);
\draw[very thick, -] (1,1)--(1,2) (3,2)--(3,1) ;
\draw[very thick, -, blue] (1,2)--(1,3)--(3,2) (1,3)--(4,2);
\draw[very thick, -, red] (1,2)--(3,3)--(3,2) (3,3)--(4,2);
\draw[very thick, -] (4,2)--(4,1)--(5,2);
\draw[very thick, green] (2,1.5)--(5,2);
\node [below] at (2,0) {$v$};
\node [left] at (1,1) {$x$};
\node [right] at (3,1) {$y$};
\node [right] at (4,1) {$z$};
\node [right] at (5,2) {$e$};
\node [left] at (1,2) {$a$};
\node [left] at (2,1.7) {$b$};
\node [right] at (3,2) {$c$};
\node [right] at (4,2.1) {$d$};
\node [above] at (3,3) {$t$};
\node [above] at (1,3) {$r$};
\draw [fill] (1,1) circle (0.07); \draw [fill] (1,2) circle (0.07);
\draw [fill] (2,0) circle (0.07); \draw [fill] (2,1.5) circle (0.07);
\draw [fill] (1,3) circle (0.07); \draw [fill] (3,1) circle (0.07);
\draw [fill] (3,2) circle (0.07); \draw [fill] (3,3) circle (0.07);
\draw [fill] (4,1) circle (0.07); \draw [fill] (4,2) circle (0.07);
\draw [fill] (5,2) circle (0.07);
\end{scope}
\end{tikzpicture}
\end{center}
\caption{~~$N=1$.\quad (i) One doubly-augmented hexagon. (ii) A possible situation.}
\label{1aughex}
\end{figure}

%

\paragraph{\em Case $N=0$.\ No 4-cycles at $v$.}
In this case $i=0$, as $v$ is not in any 4-cycle.
First suppose that there are no edges between elements of $N_2(v)$.
Then $s=|\mathcal{Q}_v| = 12$.  We have already observed that the
maximum possible value of $j$ is $j=6$, as illustrated in Fig.~\ref{6sah}.
Since $6\leq 16 = 2s/3$ it follows that (\ref{constrainij}) holds.

 \begin{figure}[H]
\begin{center}
\begin{tikzpicture}
\draw[very thick, -] (4,0)--(0,1)--(-1,2) (0,1)--(1,2);
\draw[very thick, -] (4,0)--(4,1)--(3,2) (4,1)--(5,2);
\draw[very thick, -] (4,0)--(8,1)--(9,2) (8,1)--(7,2);
\draw[very thick,red] (-1,2)--(1,4)--(3,2) (7,2)--(1,4);
\draw[very thick,blue] (-1,2)--(3,4)--(3,2) (7,2)--(3,4);
\draw[dotted] (1,2)--(5,4)--(5,2) (9,2)--(5,4);
\draw[dotted] (1,2)--(7,4)--(5,2) (9,2)--(7,4);
\node [left] at (-1,2) {$a$};
\node [left] at (3,2) {$c$};\node [left] at (6.8,1.9) {$e$};
\node [left] at (1,4.1) {$r_1$};\node [left] at (2.9,4.1) {$s_1$};
\node [left] at (1,2) {$b$};
\node [left] at (5,2.1) {$d$};\node [right] at (9,2) {$f$};
\node [below] at (4,0) {$v$}; \node [left] at (0,1) {$x$};
\node [left] at (4,1) {$y$};\node [right] at (8,1) {$z$};
\node [left] at (5,4.1) {$r_2$};\node [left] at (6.9,4.1) {$s_2$};
\draw [fill] (-1,2) circle (0.07); \draw [fill] (0,1) circle (0.07);
\draw [fill] (1,2) circle (0.07); \draw [fill] (1,4) circle (0.07);
\draw [fill] (3,2) circle (0.07); \draw [fill] (3,4) circle (0.07);
\draw [fill] (4,0) circle (0.07); \draw [fill] (4,1) circle (0.07);
\draw [fill] (5,2) circle (0.07); \draw [fill] (5,4) circle (0.07);
\draw [fill] (7,2) circle (0.07); \draw [fill] (7,4) circle (0.07);
\draw [fill] (8,1) circle (0.07); \draw [fill] (9,2) circle (0.07);
\end{tikzpicture}
\end{center}
\caption{~~$N=0$:\quad Example with $j=6$ singly-augmented hexagons centred at $v$.  }\label{6sah}
\end{figure}

Next, suppose that $N_2(v)$ induces at least one edge.  Let $a,b$ be
neighbours of $x$, $c,d$ be neighbours of $y$ and $e,f$ be neighbours of $z$,
with $a,b,c,d,e,f$ all in $N_2(v)$.  Since $N=0$, these 6 vertices in $N_2(v)$ are all distinct.

If $j=0$ then (\ref{constrainij}) holds trivially, so we assume that $j\geq 1$.
Without loss of generality, assume that there is a 4-cycle $arcta$
with $r,t\in N_2(v)\cup N_3(v)$.  We claim that $s\geq 7$. To prove this, we consider
cases for $r,t$ and fill in edges within $N_2(v)$ to make $s=|\mathcal{Q}_v|$  as small as possible.

\begin{itemize}
\item First suppose that $r,t\in N_3(v)$. Then at most 4 of the 12 possible path-pairs from $v$
can be ``blocked'' by adding edges between the remaining vertices $b,d,e,f\in N_2(v)$ of
free degree two.  For example, adding the 4-cycle $bedfb$ blocks the path-pairs $(b;(e,f))$, $(d;(e,f))$.
Hence $s\geq 8$.
\item Now suppose that $r\in N_3(v)$ and $t\in N_2(v)$, say.  To minimise $s$ we may assume without
loss of generality that $t=e$.  The 4-cycle $arcea$ blocks two path-pairs, and a further 3 path-pairs
can be blocked by adding edges between the vertices $b,d,f$, for example by adding the 3-cycle $bdfb$.
Hence $s\geq 12-5 = 7$.
\item  If $r,t\in N_2(v)$, say $r=f$, $t=e$, then the edges $af$, $fc$, $ce$, $ea$ block 4 path-pairs.
The remaining vertices $b,e$ have free degree 2, but can only be used to block one additional path-pair
(by adding the edge $be$).  Again $s\geq 7$.
\end{itemize}

To complete the proof, it suffices to show that $j\leq 4$, as this implies that (\ref{constrainij}) holds.
Note that $j=6$ (the maximum possible value of $j$) is only obtained when there are no edges induced by
$N_2(v)$, as in Fig.~\ref{6sah}.
Adding an edge between two vertices of $N_2(v)$ reduces the number of singly-augmented hexagons centered
on $v$ by at least two, as can be checked by considering cases.   For an extremal example,
suppose that $a,c,e$ are all incident with distinct vertices $r,t\in N_3(v)$, giving three singly-augmented
hexagons, and $bd$ is an edge.  We can create at most one more singly-augmented hexagon centered
on $v$, using a 4-cycle $bdgfb$ where $g\in N_3(v)$ is distinct from $r,t$.

This completes the proof  of Lemma~\ref{lem:alpha}.

\bigskip

\subsection{The analysis: Chain~II}

We will use the following version of Hoeffding's Lemma, see e.g., Theorem~5.7 of \cite{Colin}.
Let $X_1,...,X_t$ be independent random variables
which satisfy $\a \le X_k \le \b$ for $k=1,\ldots, t$.
Define the random variables $X=\sum_{k=1}^t X_k$ and $\overline X=X/t$ and let $p=\E \overline X$. Then
\begin{equation}\label{hoff}
\Pr(\overline X-p \le -\ell) \le \exp(-2 t\ell^2/(\b-\a)^2).
\end{equation}

\medskip

The following lemma will also provide a proof of Theorem~\ref{TH2}.

\begin{lemma}\label{Drift}
For $t \le \tau$, let $\D(t)$ be the number of triangles in $G(t)$ at step $t$ of the triangle process.
There are constants $0< a \le  b < 1$, such that the following holds: if $\tau \geq Cn$ for some sufficiently large constant $C>0$
then for all sufficiently small $\epsilon>0$, w.h.p.\ $\D(t) \ge a(1-\epsilon)n$ for all but $o(\tau)$ steps and $\D(t) \le b(1+\epsilon)n$ for all but $o(\tau)$ steps.
\end{lemma}

\begin{proof}
We focus the lower bound, as the proof for the upper bound is similar (see remarks at the end of the proof.)
Let $X_k$ be the number of triangles created at step~$k$ of the process (note that $X_k$ may be negative),
and let $X(t)=\sum_{k=0}^t X_k$
 be the total number of triangles, where $X_0$ is the number of triangles in the initial graph.
Thus $X(t)$ is a random walk on $\{0,1,...,n\}$, with reflecting barriers.
The walk has
step sizes bounded by  $\{-4,...,4\}$,  so $\a=-4$, and $\b=4$.

For any graph $G$ with $ns$ vertices in triangles, the function $f(s)=3-10s-10s^2/3$ from \eqref{MusLB}
is a lower bound on the expected change in the number of vertices in
triangles at one step.  Let $\d(s)=f(s)/3$.
The function $\d(s)$ has $\d(0)=3$ and is monotone decreasing in $s$ until some value $s^+$, when $\d(s^+)=0$. Let $a=s^+/3>0.0916$. See below \eqref{MusLB} for more detail.
  For $\D=j$, and $j < an$, if there are $j$ triangles then  $\E X_k \ge \d(j)>0$, and for $i<j\le a$, $\d(i)>\d(j)$.


For  $\eta>0$ let $c''=a(1-3\eta), c'=a(1-2\eta),
c=a(1-\eta)$.
Starting from no triangles (which is the worst case assumption),
 we show that w.h.p.\ the walk passes $cn$ in at most $t=\Theta(n)$ steps.

Suppose that $X(t) <cn$. Then $\overline X=X(t)/t \ge \d=\d(cn)$.
Using \eqref{hoff},
\[
\Pr(X(t)<cn) \le \Pr(X \le t\d-t\ell) \le \exp(- t\ell^2/32).
\]
With $\ell=\d/2$ and $t=(2cn+1)/\d$, we find that
\[
\Pr(X(t)<cn) \le \exp(-\d a(1-\eta)n/64).
\]
Thus either the walk passes $cn$  in $t<3cn/\d=\Theta(n)$ steps or an event of probability
$e^{-\Theta(n)}$ occurs.
Let $\eta$ above be such that $c=0.09=\a$ from Theorem~\ref{TH2}(i).
This proves the lower bound in Theorem~\ref{TH2}(i).


Next, suppose that the walk returns to $c'n$. We restart our counting, and set $X_0=c'n$ at this point. Let $\mathcal{A}_t$ be the event that we reach $c''n$ in a further $t$ steps, without returning to $cn$, and let $\mathcal{A}=\cup \mathcal{A}_t$.
Note that $|c'-c''|=|c-c'|=a\eta$. As we can break at most 4 triangles at any step we have $t \ge a \eta n/4$. The event $\mathcal{A}_t$ is equivalent to $X(t) \le -a\eta n$.
Take $\ell=\d + a\eta n/t$, then
\[
\Pr(\mathcal{A}_t)=\Pr(X(t) \le -a\eta n)\le\Pr(X(t) \le t\d-t\ell)\le \exp(- t\ell^2/32),
\]
Thus
\[
\sum_{t \ge a\eta n/4} \Pr(\mathcal{A}_t) \le \sum_{t \ge a\eta n/4}\exp(-t\d^2/32)=e^{-\lambda n},
\]
for some constant $\lambda>0$, i.e., if the number of triangles falls to $c' n$,   the  probability  to go as low as $c'' n$ before climbing back up to $cn$
is exponentially small.
As every departure from $cn$ downward is a renewal process, the probability of reaching $c''n$ in $e^{\lambda n/2}$ such excursions is exponentially small.
The proof of the upper bound is similar, with $b=0.628$.
Observe that over linearly many steps, the $O(1/n)$ additive error in (\ref{upper-quad})
can only contribute an $O(1)$ additive error on the number of triangles,
which is asymptotically insignificant. So we may ignore this error term in our
analysis.
This completes the proof of the lemma and establishes Theorem~\ref{TH2}(ii).
\end{proof}


\subsection{Experimental evidence}
Fig.~\ref{fig:sim} shows a simulation of Chain II on a 3-regular graph of with
$n=2000$ vertices.
\begin{figure}[ht!]
\begin{center}
\includegraphics[trim= 10mm 15mm 0mm 15mm, clip,  height=0.35\textheight]{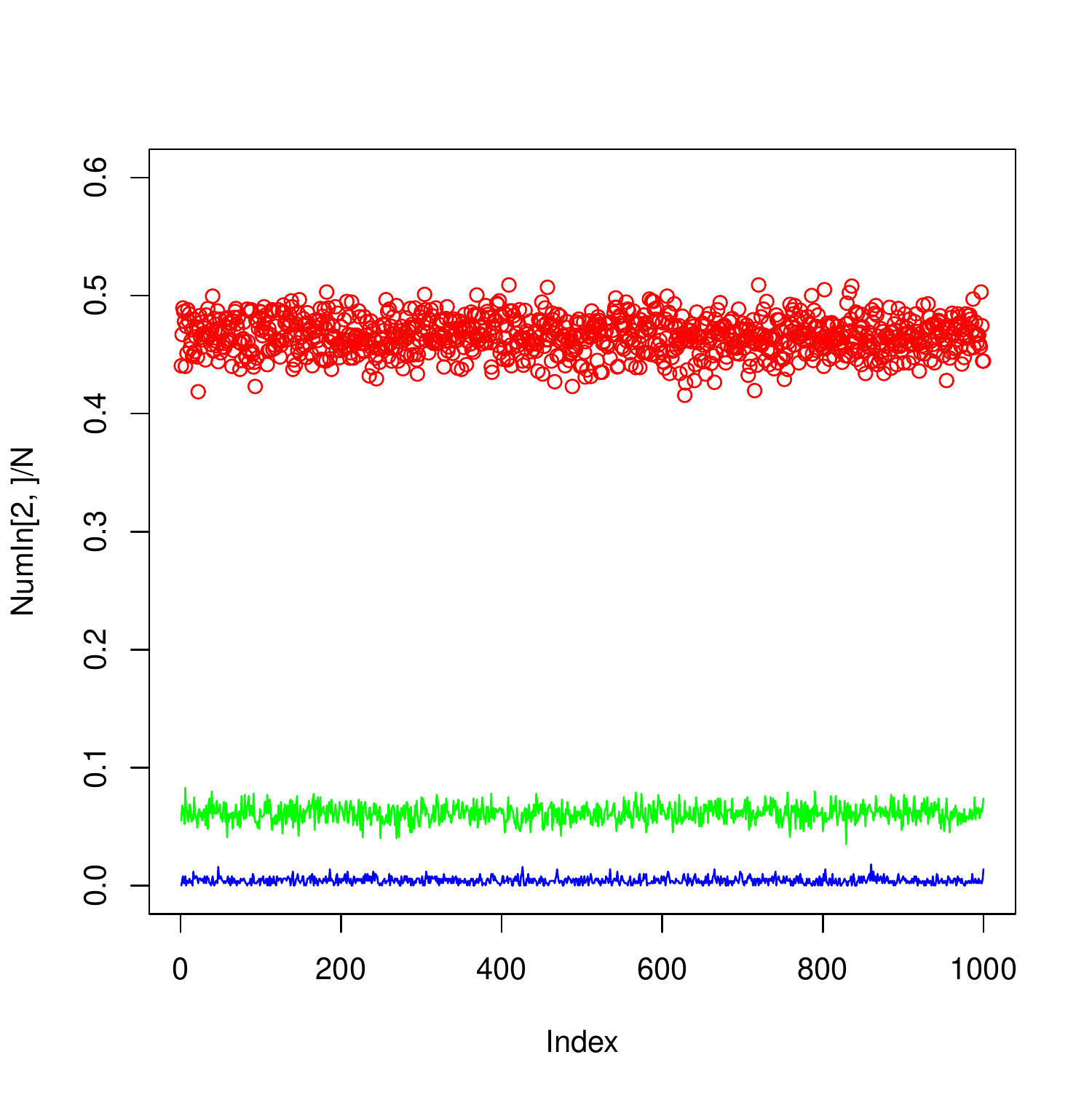}
\vspace*{\baselineskip}
\caption{\,\, Fraction of vertices in isolated triangles (top), diamond diagonals (middle) and tetrahedra (bottom), over
a simulation with a million steps}
\label{fig:sim}
\end{center}
\end{figure}
The scale of the horizontal axis is one for every 1000 simulation steps: that is,
the simulation ran for one million steps. The vertical axis is the fraction of vertices of the following types:
\begin{itemize}
\item vertices on isolated triangles  (red, the most common);
\item vertices on ``diamond diagonals'': that is, the two vertices incident with both triangles of a diamond (green);
\item vertices on tetrahedra  (blue, the least common).
\end{itemize}
\ignore{
Denote these fractions by $\ell$, $r$ and $t$, respectively.  Then the
total number of triangles is
\[ \frac{(\ell +3r+3t)n}{3},\]
as each isolated triangle has 3 vertices, each diamond has 2 vertices on the diagonal and 2 triangles, and
each tetrahedron has 4 vertices and 4 triangles.
}
The average number of triangles during such
simulations was approximately $0.2 n$.  Fig.~\ref{fig:sim2}
shows the number of triangles in a simulation of Chain~II on an $n= 1,000$ vertex graph, the average number being 200.1365 ($0.2001365 n$). Note that the number of triangles never falls below 90 ($0.09n$), see Lemma~\ref{Drift}.
\begin{figure}[ht!]
\begin{center}
\includegraphics[trim= 10mm 15mm 0mm 15mm, clip,  height=0.35\textheight]{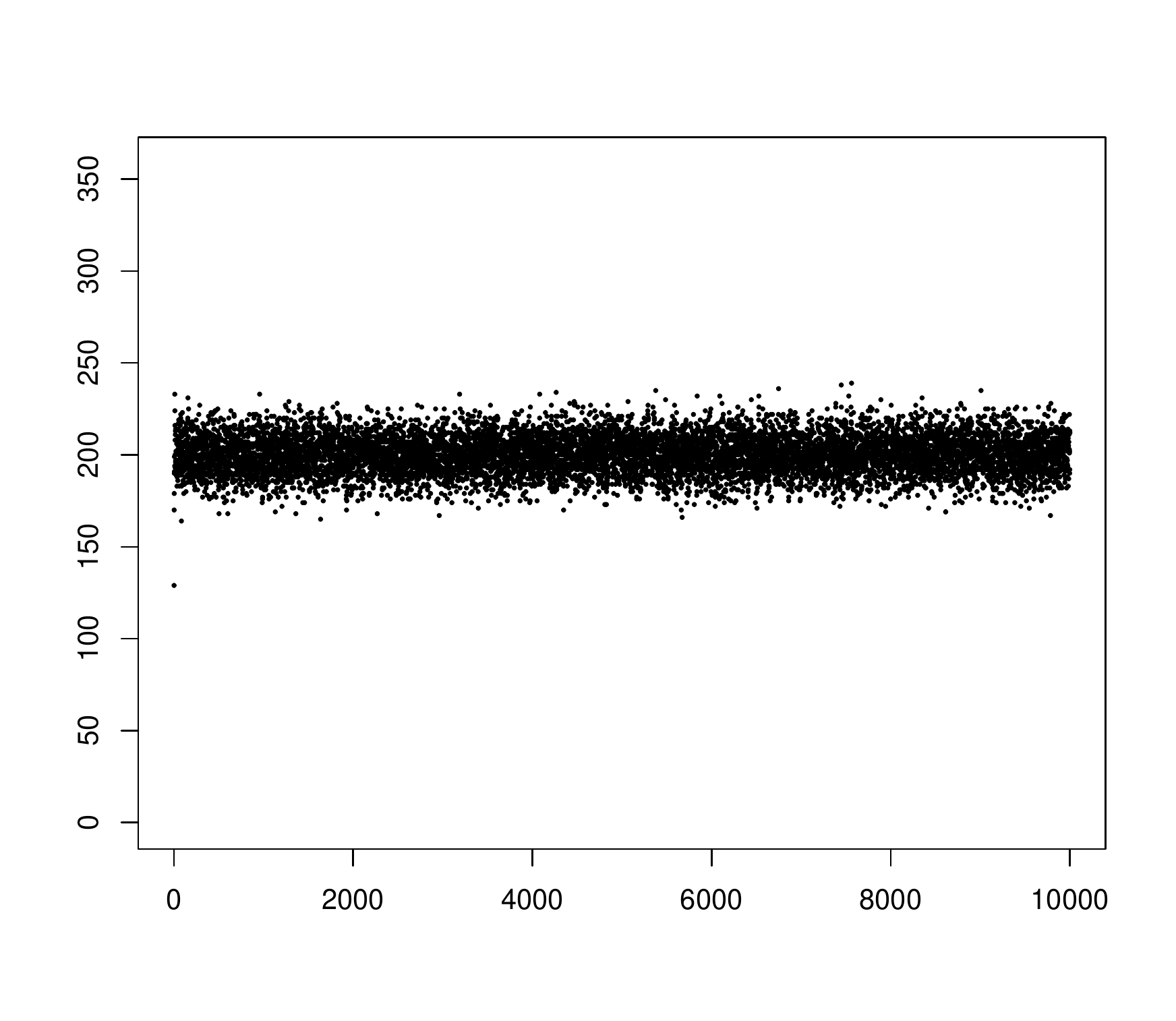}
\vspace*{\baselineskip}
\caption{\,\, Number of triangles in a 1,000 vertex graph for 10,000 samples taken over 2 million simulation steps}
\label{fig:sim2}
\end{center}
\end{figure}

\subsection{Adapting the analysis to Chain I}\label{number:secCI}

During a \tm at a free vertex, Lemma~\ref{ergo:lem10} guarantees the
existence of a valid move. As there are at most 12 paths of length 2 starting from $v$, a
\tm\ is successfully performed with probability at least $p/12$.
We adapt the derivation of inequality \eqref{MusLB} to give
\[
\mu_S \ge \frac{p}{12}(3-{6s})- (1-p)(4s+10s^2/3),
\]
where $s$ is the number of vertices belonging to at least one triangle.
We find that $\mu_S>0$ holds when 
\begin{equation}
\label{quadratic}
f(s)= s^2 10(1-p)/3+s(4-7p/2)-p/4 <0.
\end{equation}
We treat the left hand side as a quadratic in $s$, and note that $f(0)<0$.
Using the fact that $\sqrt{1+x}\geq 1 + x/3$ for $0 \le x \le 3$, we find that the positive root of $f(s)$ in  (\ref{quadratic}) is greater than $s^+=p/(24-21p)$. As $\D \ge ns/3$, it follows from Lemma~\ref{Drift} that for any $\e >0$ constant, a lower bound on the long run number of triangles is $\D=(1-\e)np/(72-63p)$. This proves Corollary \ref{coTH2}.

\section{Conclusions}
We have proposed and analysed simple Markov chains for randomly generating 3-regular graphs with a large number of triangles.
A natural question is whether this can be generalised to $d$-regular graphs for $d>3$.
Proving irreducibility by the kind of case-analysis used in Section~\ref{ergo:sec} appears  problematic:
even generalising Lemmas \ref{ergo:lem10}--\ref{ergo:lem30} may be difficult.
Furthermore, the ``left-over'' graph (the last component after most of the graph has been turned into copies of
$K_{d+1}$) will have $(d+1)+(n\bmod d+1)$ vertices, and the number of
non-isomorphic graphs of this order increases quickly as $d$ increases.



Proving rapid mixing of the triangle switch chains discussed here would seem to be difficult without a more canonical proof of irreducibility. Such a proof may be possible for graphs restricted to
have $n-o(n)$ triangles, by exploiting the simple structure of $\K^*_n$.
It may even be that triangle switch chains
are slow mixing in general for some $d \ge 3$.

\end{document}